\pdfoutput=1
\documentclass{article}

\usepackage[english]{babel}
\usepackage[margin=1in]{geometry}
\usepackage{amsmath, amssymb, textcomp}
\usepackage{graphicx}
\usepackage{float}
\usepackage{booktabs}
\usepackage{threeparttable}
\usepackage{tabularx}
\usepackage{array}
\newcolumntype{Y}{>{\centering\arraybackslash}X}
\usepackage{amsthm}
\theoremstyle{plain}
\newtheorem{proposition}{Proposition}[section]

\newtheorem{heuristic}{Heuristic}[section]
\newtheorem{lemma}{Lemma}[section]
\newtheorem{theorem}{Theorem}[section]
\newtheorem{corollary}[theorem]{Corollary}
\newtheorem{heuristicproposition}[proposition]{Heuristic Proposition}
\theoremstyle{remark}
\newtheorem*{remark}{Remark}
\theoremstyle{definition}
\newtheorem{definition}{Definition}[section]
\usepackage{tikz}
\usepackage{tikz-3dplot}
\usetikzlibrary{arrows.meta,positioning,fit,shadows.blur}
\usepackage{caption}
\usepackage{subcaption}
\usepackage[numbers]{natbib}

\newcommand{\tmaffiliation}[1]{\\ #1}
\newcommand{\tmem}[1]{{\em #1\/}}
\newcommand{\tmemail}[1]{\\ \textit{Email:} \texttt{#1}}
\newcommand{\tmop}[1]{\ensuremath{\operatorname{#1}}}

\newcommand{\tmstrong}[1]{\textbf{#1}}
\newcommand{\tmtextbf}[1]{{\bfseries{#1}}}
\newcommand{\tmtextit}[1]{{\itshape{#1}}}

\newenvironment{tmparmod}[3]{
  \begin{list}{}{\setlength{\topsep}{0pt}
  \setlength{\leftmargin}{#1}
  \setlength{\rightmargin}{#2}
  \setlength{\parindent}{#3}
  \setlength{\listparindent}{\parindent}
  \setlength{\itemindent}{\parindent}
  \setlength{\parsep}{\parskip}}
  \item[]
}{\end{list}}

\usepackage[hyphens]{url}
\usepackage[hidelinks,breaklinks=true,bookmarksopen=true,pdfencoding=auto]{hyperref}
\hypersetup{
  pdftitle    = {Symbolic Constraints in Polyhedral Enclosure and Tetrahedral Decomposition in Genus-0 Polyhedra},
  pdfauthor   = {Moustapha Itani},
  pdfkeywords = {polyhedral combinatorics, Euler characteristic, triangulations}
}

\begin{document}

\title{Symbolic Constraints in Polyhedral Enclosure and Tetrahedral Decomposition in Genus-0 Polyhedra}

\author{
  Moustapha Itani
  \tmaffiliation{ORCID: 0000-0001-9372-3715}
  \tmemail{moustapha.itani@helsinki.fi}
}

\date{September 29, 2026}

\maketitle

\begin{abstract}
I present an exploratory, coordinate-free framework for organizing two related
problems in polyhedral combinatorics: screening polygonal face multisets for
genus-zero enclosure and examining tetrahedral decompositions that preserve the
original boundary vertex set.

The boundary is summarized using the face-degree excess
\[
S := \sum_f (\deg f-3),
\]
which equals the number of diagonals required to triangulate all polygonal faces
without adding vertices. Standard Euler and incidence relations then provide
exact relations among the numbers of vertices, edges, faces, and boundary
triangulation diagonals, together with parity-dependent necessary bounds derived
from the minimum vertex degree.

For tetrahedral decompositions, standard incidence relations for a triangulated
$3$--ball are used to separate general topological constraints from the
additional structure generated by the restricted boundary-driven decomposition
procedure studied here. Within the manually constructed examples, bipyramidal
configurations provide an extremal reference class, while departures from
bipyramidal structure are associated with changes in attainable tetrahedral
states. The examples further suggest a contraction of the realized internal
state space as boundary face-degree excess increases, and motivate the notions
of $S$--flexibility and incidence-gated transitions. These latter patterns are
presented as empirical structural observations and conjectural directions rather
than as general classification theorems.

These ingredients are combined into a symbolic pre-screening workflow that can
reject some impossible polygonal face multisets before explicit graph
construction or geometric realization. The framework provides necessary
combinatorial conditions only. An open-source R implementation is available at
\url{https://github.com/MoustaphaItani/polyenclose/}.
\end{abstract}

\noindent\textbf{Author's Note.}
This work began with polyhedra I drew and counted in my early teenage years, long before I knew how to describe what I was doing in formal mathematical language. The angle-unit bookkeeping, constructions, and decomposition ideas in this manuscript grew from that work. When I returned to these ideas and began formalizing them many years later, I found that several of the relations I had recovered were standard consequences of classical Euler and incidence theory. I therefore treat those relations here as mathematical background rather than as claims of novelty.

One reason for placing this work on arXiv is simply to leave a public record of the ideas and methods in the form I have been able to develop them. If I am unable to take them further, they remain available for someone else, or some computational system, to examine, correct, extend, use, or ignore. I have tried to document them as carefully and honestly as I can with the training, time, and resources available to me.

The present arXiv submission should be read as a working exploratory document and, in part, as a public drawing board rather than as a finished mathematical treatment. Preparing it has also been a way of learning how to express earlier informal explorations in formal mathematical language and organize them in \LaTeX. AI-assisted tools have been used for notation, coding assistance, derivation checks, literature exploration, and editing. They were not the source of the underlying constructions or exploratory ideas. The author remains responsible for the mathematical claims and classifications in the present version.

\section{Introduction and Background}

This paper documents an exploratory symbolic framework for reasoning about
polyhedral boundary structure and restricted tetrahedral decomposition without
beginning from spatial coordinates. The framework grew from manual examination
of polyhedral forms and from tracking how vertices, edges, faces, boundary
triangulations, and tetrahedral fillings changed under simple combinatorial
modifications.

The central boundary quantity used throughout is
\[
S := \sum_f (\deg f-3),
\]
the total face-degree excess relative to triangular faces. Equivalently, $S$ is
the number of non-intersecting diagonals required to triangulate all boundary
faces without introducing new vertices. I use $S$ as a convenient discrete
coordinate for organizing boundary configurations.

\paragraph{Scope.}
The paper concerns genus-zero polyhedral surfaces whose boundary graphs are
simple and $3$--connected where polyhedral realizability is invoked. For the
internal problem, attention is restricted primarily to tetrahedralizations with
no interior vertices and no edge bisections. A more restrictive
boundary-driven decomposition procedure is introduced later for the empirical
part of the study.

The framework does not provide a complete realizability criterion for an
arbitrary face multiset. Rather, it supplies a sequence of necessary
combinatorial checks that can reject some impossible candidates before explicit
graph construction or geometric realization.

\paragraph{Mathematical background and contribution.}
Several exact identities used in the paper are standard consequences of Euler's
formula, face--edge incidence, the handshake lemma, and the Euler
characteristic of a triangulated $3$--ball. In particular, the relations among
$V$, $E$, $F$, and $S$, the parity-dependent edge bounds, and the basic
tetrahedral incidence identities are used here as mathematical infrastructure,
not as claims of priority.

The contribution of the present work lies not in the classical incidence
identities themselves, but in the framework built around them. This includes
the angle-unit bookkeeping through which the problem was originally developed;
the use of the face-degree excess $S$ as a common coordinate linking boundary
structure to symbolic screening; the restricted boundary-driven tetrahedral
decomposition procedure; and the structural patterns identified through manual
construction. In particular, the latter include the role of bipyramidal
configurations as an extremal reference class, the apparent contraction of
realized tetrahedral states as boundary face-degree excess increases, and the
resulting questions of $S$--flexibility and incidence-gated transitions.

Some of these structural observations may overlap with established results in
the triangulation and polyhedral-combinatorics literature, while others appear
to constitute specific features or formulations of the framework developed
here. The paper therefore distinguishes established identities from
procedure-dependent observations and from claims whose broader novelty or
generality remains to be established.

Because this work has been developed without formal mathematical supervision, its engagement with the existing literature is incomplete. Some statements may be insufficiently cited, and some claims may unintentionally overstate their originality. Claims of novelty should therefore be regarded as provisional unless explicitly supported by comparison with prior work. Later versions may correct such issues as the mathematical and bibliographic context becomes clearer.

The paper therefore separates three levels of analysis:
\begin{enumerate}
    \item standard boundary and volumetric incidence identities;
    \item a symbolic screening procedure based on necessary combinatorial
          conditions; and
    \item empirical observations generated by the restricted decomposition
          procedure.
\end{enumerate}
\clearpage
Steinitz's theorem provides the classical graph-theoretic characterization of
convex polyhedral graphs: a graph is the $1$--skeleton of a convex
$3$--polyhedron if and only if it is planar and $3$--connected
\cite{steinitz1922}. Standard references on planar graphs, polytopes,
triangulations, and Euler--Poincar\'e relations provide the broader mathematical
setting for the identities used below
\cite{ziegler2000,gallier2022aspects}.

To ensure clarity, Table 1 summarizes the key variables and
symbolic terms used throughout this document. These variables describe
either the external surface of a polyhedron or structures required for its
internal decomposition into tetrahedra.

\begin{table}[H]
  \centering
  \caption{Key variables used in the symbolic framework. 
  \emph{Surface/external} variables describe the boundary structure, while \emph{internal} variables capture the combinatorial scaffolding of tetrahedral decompositions. 
  This grouping highlights the two structural layers that the framework relates through symbolic identities.}

  \begin{tabular}{lp{12cm}}
    \hline
    \textbf{Symbol} & \textbf{Meaning} \\
    \hline
    \multicolumn{2}{l}{\textbf{\textit{Surface / External Structure Variables}}} \\
    $V$ & Number of vertices \\
    $E$ & Number of edges \\
    $F$ & Number of faces \\
    $S$ & External non-intersecting triangulation segments (used to triangulate surface); also used throughout as a discrete flatness measure \\
    $N$ & Total internal angle units on the external surface (1 unit $=\pi$ radians $=180^\circ$) \\
    $F_\triangle$ & Number of boundary triangular faces resulting from the triangulation of the entire polyhedral surface \\
    $M$ & External angle fullness units (derived from polygon types) \\
    $F_3, F_4, F_5,\ldots$ & Face types: triangle (3 sides), quadrilateral (4), pentagon (5), etc. \\
    $V_3, V_4, V_5,\ldots$ & Vertex types: valency 3 (3 incident edges), valency 4, valency 5, etc. \\
    \hline
    \multicolumn{2}{l}{\textbf{\textit{Internal Decomposition Variables}}} \\
    $T$ & Number of internal tetrahedra (used to fill the volume) \\
    $N_i$ & Number of internal gluing angle units (each unit $=$ one triangle) \\
    $S_i$ & Number of distinct interior triangulation segments introduced during
    the decomposition. Each segment is counted once when first introduced and may
    subsequently belong to several internal triangular faces. In normal form these
    segments are precisely the interior edges, so $S_i=E_i$ \\
    $E_i$ & Number of interior edges of the resulting tetrahedral complex; in
    normal form $E_i=S_i$ \\
    $V_i$ & Internal vertices \\
    \hline
  \end{tabular}
\end{table}

The variables above are not derived from physical measurements or coordinates,
but from symbolic relationships among discrete structural components. For
example, angle units are treated symbolically as multiples of $\pi$, and
surface structures are quantified via segment and gluing unit counts rather than metric angles or lengths.

The auxiliary variables $N$, $M$, $S$, $N_i$, and $S_i$ do not appear in the classical Euler–Steinitz framework.
They are introduced here to support symbolic decomposition bookkeeping and to connect surface composition to volumetric tetrahedralization. 
$N$ represents the total angle units on the external surface, with one unit defined as $\pi$ radians ($180^\circ$), corresponding to the interior-angle sum of one triangle.
$M$ is derived from polygon types and measures “external angle fullness” in the same units. $S$ counts the non-intersecting triangulation
segments added to polygonal faces to produce a fully triangulated surface and also serves as a discrete flatness measure.
$N_i$ counts the number of internal gluing angle units, each corresponding to a shared triangular face between two tetrahedra. 
$S_i$ counts the internal non-intersecting triangulation segments required to complete a volumetric mesh during decomposition.
These auxiliary quantities are introduced as bookkeeping variables, although
several have direct combinatorial realizations. In particular, $S$ counts
boundary triangulation diagonals, $N_i$ counts internal triangular gluing
faces, and $S_i$ counts the distinct interior edges introduced during the
decomposition. The latter are counted once at their first construction,
irrespective of how many internal faces subsequently contain them.

\subsection{Illustrative Examples of $N_i$, $S$, and $S_i$}

\paragraph{\tmtextbf{Rectangular Pyramid ($S = 1$):}} A five-vertex polyhedron
with a quadrilateral base and one apex point. To triangulate the external
surface, a single diagonal segment divides the base into two triangles,
yielding $S = 1$.

\begin{figure}[H]
\centering
\tdplotsetmaincoords{70}{120}
\begin{tikzpicture}[tdplot_main_coords,line cap=butt,line join=round,
    c/.style={circle,fill,inner sep=1pt},
    declare function={a=4;h=3;}]
    \path
    (0,0,0) coordinate (A)
    (a,0,0) coordinate (B)
    (a,a,0) coordinate (C)
    (0,a,0) coordinate (D)     
    (a/2,a/2,h)  coordinate (S);
    \draw[orange, thick] (B) -- (D);
    \draw (S) -- (D) -- (C) -- (B) -- cycle (S) -- (C);
    \draw[dashed] (S) -- (A) -- (D) (A) -- (B);
\end{tikzpicture}
\caption{Rectangular pyramid illustrating the flatness measure \(S\).
The single internal diagonal on the quadrilateral base (orange) is the unique surface triangulation segment needed to obtain a fully triangulated boundary, hence \(S=1\).
This corresponds to one coplanarity merge relative to the maximally triangulated case and, by Proposition~\ref{prop:triangulated-extrema}, decreases both \(F\) and \(E\) by \(1\) while preserving \(V\) and \(\chi=2\).}
\label{fig:rect-pyramid-S1}
\end{figure}
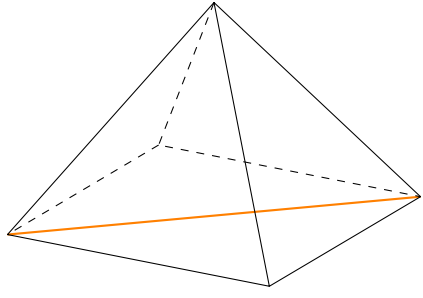

\paragraph{\tmtextbf{Triangular Bipyramid ($N_i = 1$):}}Composed of two
tetrahedra joined along a shared triangular face, this structure has one
gluing triangle (blue): $N_i = 1$.

\begin{figure}[H]
\centering
\tdplotsetmaincoords{70}{120}
\begin{tikzpicture}[tdplot_main_coords, line cap=butt, line join=round]
    \path
    (0,0,0) coordinate (A)
    (4,0,0) coordinate (B)
    (2,2.309,0) coordinate (C)
    (2,2.309,4) coordinate (S)
    (2,2.309,-4) coordinate (N);
    
    \fill[blue, opacity=0.3] (A) -- (B) -- (C) -- cycle;
    \draw[dotted] (A) -- (B) -- (C) -- cycle;
    \draw (B) -- (C);
    \draw[dotted] (S) -- (A) (A) -- (C) (A) -- (B) (N) -- (A);
    \draw (S) -- (B) (S) -- (C);
    \draw (N) -- (B) (N) -- (C);
\end{tikzpicture}
\caption{Triangular bipyramid obtained by gluing two tetrahedra along a common triangular face (highlighted in blue). 
This shared face corresponds to exactly one internal gluing triangle, hence $N_i = 1$ in the symbolic model.
It provides the simplest example of how $N_i$ is counted: each such internal face removes two boundary triangles in the fully triangulated representation,
linking interior adjacency directly to external face counts via Proposition~\ref{prop:boundary-from-interior}.}
\end{figure}
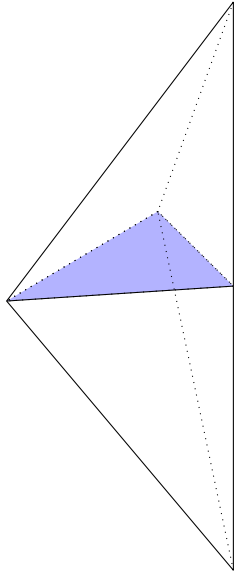

\paragraph{\tmtextbf{Irregular Octahedron ($S_i = E_i = 1$, $N_i = 4$):}} When
regular, constructed by gluing two rectangular pyramids base-to-base. In the
case of this irregular octahedron, the internal meshwork between the glued
tetrahedra must be constructed by adding one internal segment functioning as an internal edge, yielding $S_i = E_i =1$.
Once included, four gluing triangles would emerge, resulting in $N_i = 4$.

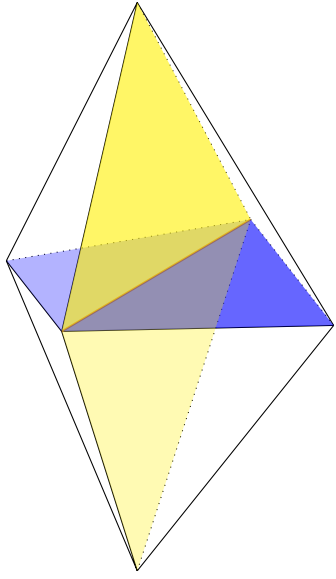
\begin{figure}[H]
\centering
\tdplotsetmaincoords{70}{120}
\begin{tikzpicture}[tdplot_main_coords, line cap=butt, line join=round]
    \path
    (0,0,0) coordinate (O)
    (2,0,0) coordinate (A)
    (-3,0,0) coordinate (B)
    (0,3,0) coordinate (C)
    (0,-2,0) coordinate (D)
    (0,0,4) coordinate (S)
    (0,0,-4) coordinate (N);
    
    \draw (A) -- (C) (D) -- (A);
    \draw[orange, thick] (A) -- (B);
    \draw (S) -- (A) (S) -- (C) (S) -- (D);
    \draw[dotted] (D) -- (B) (N) -- (B) (B) -- (C) (S) -- (B);
    \draw (N) -- (A) (N) -- (C) (N) -- (D);
    \fill[blue, opacity=0.6] (A) -- (B) -- (C) -- cycle;
    \fill[blue, opacity=0.3] (D) -- (A) -- (B) -- cycle;
    \fill[yellow, opacity=0.6] (S) -- (A) -- (B) -- cycle;
    \fill[yellow, opacity=0.3] (N) -- (A) -- (B) -- cycle;
\end{tikzpicture}
\caption{Irregular octahedron constructed by gluing two quadrilateral-based pyramids along their bases. 
The orange segment is an internal triangulation segment ($S_i = 1$) required to fully mesh the internal volume into tetrahedra without adding new vertices. 
This configuration produces four internal gluing triangles ($N_i = 4$), shown in blue and yellow, illustrating how $S_i$ and $N_i$ together describe the combinatorial scaffolding of the internal decomposition.}

\end{figure}

\section{Preliminaries}\label{sec:prelim}

\subsection{Normal form and standing assumptions}\label{subsec:normal-form}

\begin{definition}[Normal form]\label{def:normal-form}
A genus-$0$ polyhedral tetrahedralization is in normal form in this paper if:
\begin{enumerate}
    \item all vertices lie on the boundary, so no interior/Steiner vertices
    are introduced and $V_i\equiv0$;
    \item boundary and interior edges are straight segments and distinct edges
    do not intersect in their relative interiors;
    \item every interior triangulation segment joins two existing boundary
    vertices and is counted once when first introduced;
    \item distinct interior triangulation segments may share endpoints but do
    not cross one another; and
    \item tetrahedra meet face-to-face along triangular interfaces, with no
    edge bisections.
\end{enumerate}
\end{definition}

\begin{remark}[Interior faces and edges]
An internal triangular face is shared by exactly two tetrahedra. An interior
edge, by contrast, may be incident to several internal triangular faces and
several tetrahedra. Once an internal triangulation segment has been introduced,
subsequent faces may use that same segment without increasing $S_i$.
Under the normal-form convention, the internal triangulation segments are
precisely the interior edges, hence
\[
S_i=E_i.
\]
\end{remark}
\clearpage
\section{Methods and Heuristics}

This section introduces both the \emph{proved combinatorial identities} and the \emph{empirical heuristics}
that underlie the framework, along with the triangulation logic that directs it towards constructing
a polyhedron's internal mesh.

\subsection{Logical separation of results}

Throughout the paper I distinguish between:

\begin{itemize}
    \item \textbf{Exact identities}, derived from Euler characteristics and
    incidence counting, including
    \[
    E=3V-6-S,\qquad F=2V-4-S,
    \]
    and, for the tetrahedralized $3$--ball,
    \[
    T-N_i+E_i-V_i=1.
    \]

    \item \textbf{Construction-dependent observations}, including proposed
    upper ranges, realized configuration sets, bipyramidal reference states,
    and empirical relationships between boundary structure and internal
    decomposition.
\end{itemize}

\subsection{Empirical construction procedure}

Manual exploration proceeded by:
\begin{itemize}
    \item modifying one boundary or internal structural element at a time;
    \item recording the resulting values of $V,E,F,S,T,N_i,$ and $S_i$; and
    \item comparing the resulting sequences at fixed $V$ as boundary
    structure changed.
\end{itemize}

\subsection{Surface flatness estimation}

To quantify \emph{flatness} on the surface, I use the face-degree excess
\[
S:=\sum_f(\deg f-3),
\]
which equals the total number of non-intersecting diagonals needed to triangulate
all polygonal faces. Thus $S=0$ if and only if every face is triangular, and
increasing $S$ corresponds to fewer faces and edges relative to the triangulated
case via Proposition~\ref{prop:ext-formulas}.

\begin{remark}
Because the embedding is planar with all faces homeomorphic to disks
(i.e., a $2$--cell embedding of $S^2$), each edge is incident to exactly two
distinct faces.
\end{remark}

Each polygonal face is triangulated by drawing non-intersecting diagonals
between its existing vertices. The cumulative count of these added
segments on all external faces constitutes the flatness measure $S$. When
all faces of a polyhedron are triangles, $S = 0$ and $F \equiv F_{\mathrm{ub}}$.

\subsection{Internal mesh construction from surface triangulation}

Unlike the external surface of a polyhedron, which is both observable and abstractable with consistency,
the internal structure must be inferred. Consequently, the number of tetrahedra required
to tetrahedralize a polyhedron, such as a cube, depends critically on the configuration of face cuts.
For example, as noted by Timalsina and Knepley~\cite{timalsina2023tetrahedralization},
a cube may be divided into 5, 6, or 12 tetrahedra depending on the chosen strategy: five tetrahedra in an ideal face-aligned configuration, six using the marching tetrahedra algorithm, and twelve when a Steiner point is introduced at the cube’s center. The 6 and 12 tetrahedra decompositions are valid tetrahedralizations in the general sense, but they do not necessarily belong to the restricted successive-peeling family considered here.

Since the goal of this exploration is to preserve the original vertex set while minimizing additional subdivision and restructuring of the internal complex, I adopt the restricted boundary-driven strategy described below.
For this purpose, the internal mesh is constructed solely from external constraints, namely, the original vertices and triangulated faces.

To support this construction, internal triangulation segments are introduced
to define the internal gluing triangles along which tetrahedra are joined.
Each distinct internal segment contributes once to the structural count
$S_i$, at the step at which it is first introduced. The same segment may
subsequently occur in several internal triangular faces without being counted
again. Each internal triangular gluing face contributes once to $N_i$ and is
shared by exactly two tetrahedra.

\subsection{Normal-form restrictions on interior structure}

In the normal-form constructions considered here, all vertices belong to the
boundary and hence
\[
V_i=0.
\]
Every internal triangulation segment joins two existing boundary vertices.
Distinct such segments have disjoint relative interiors, although they may
share endpoints. Each segment is counted once when it first appears in the
construction.

After its introduction, an internal segment may be incident to several
triangular faces and tetrahedra. Thus the number of faces incident to an
interior edge is not restricted to two. Rather, it is the internal triangular
faces themselves that are each incident to exactly two tetrahedra.

Under these restrictions, the internal triangulation segments are precisely
the interior edges of the resulting tetrahedral complex, and therefore
\[
S_i=E_i.
\]

\subsection{Boundary-driven decomposition procedure}

The restricted decomposition used here follows five practical rules:
\begin{itemize}
    \item tetrahedra are progressively removed from the exterior inward;
    \item trivalent vertices are prioritized where their triangular link
    directly determines a tetrahedron;
    \item existing vertices and whole polygonal faces are preserved, with no
    edge bisections or Steiner vertices;
    \item a new interior segment is introduced only when required to complete the current triangular gluing;
    \item once an interior edge or triangular gluing interface has been introduced, it is retained in subsequent steps and is not subdivided, replaced, or re-triangulated.
\end{itemize}

The procedure is deliberately restrictive. It is intended to construct the internal mesh by successively identifying and removing tetrahedra while minimizing internal subdivision. In particular, the current polyhedral remainder is never repartitioned into auxiliary non-tetrahedral regions in order to obtain a different tetrahedralization.

\begin{remark}[Direct peeling and preliminary repartitioning]
The successive-peeling procedure should not be confused with unrestricted tetrahedralization. At each stage, a tetrahedron is removed directly from the current polyhedral remainder; the remainder is not first partitioned into auxiliary polyhedra merely to obtain a different tetrahedralization.

The cube illustrates the distinction. Direct successive peeling terminates
in a five-tetrahedron decomposition. A six-tetrahedron decomposition can be
obtained by first partitioning the cube into two triangular prisms and then
tetrahedralizing the prisms. 
Relative to the five-tetrahedron peeling complex, this prism partition cuts
through the structure associated with the terminal fifth tetrahedron. From
the viewpoint of the peeling sequence, the resulting six-tetrahedron
construction can therefore be understood as subdividing or retriangulating
the terminal tetrahedral region rather than as providing a further peeling
step. The additional tetrahedron thus arises through precisely the additional
internal restructuring that the peeling convention is intended to avoid.

The six-tetrahedron construction remains a valid tetrahedralization of the
cube in the general sense, but it is not an additional state generated by the
restricted successive-peeling procedure.
\end{remark}

\subsection{Validation by Euler's characteristic}

Two consistency checks are used:
\begin{itemize}
    \item \textbf{Boundary:}
    \[
    E=3V-6-S,\qquad F=2V-4-S,
    \]
    which is consistent with $V-E+F=2$.

    \item \textbf{Interior:}
    \[
    T-N_i+E_i-V_i=1,
    \]
    and in normal form,
    \[
    T-N_i+S_i=1.
    \]
    Hence
    \[
    V-E+F=2(T-N_i+S_i).
    \]
\end{itemize}
\clearpage
\section{Vertex-Driven External Polyhedral Properties: Intermediate Flatness Realizability and Parity Effects}
\label{sec:realizability}

I establish extremal bounds on external combinatorial quantities
($F$, $E$, and $S$) directly as functions of the vertex
count $V$. These bounds emerge from Euler’s characteristic combined with 
flatness-based adjustments, and serve as the foundation for linking external 
structure to the internal decomposition results in later sections. 
This section establishes the admissible parity-sensitive bounds on $E$, $F$, and $S$
and presents heuristic evidence that every integer flatness value
$0 \le S \le S_{\mathrm{ub}}$ is realizable by some genus-zero polyhedron in normal form.
The realizability result is stated formally in \emph{Proposition~\ref{prop:realizability}}.

The derivation proceeds from a given number of vertices $V \in \mathbb{N}$, 
showing how faces, edges, and triangulation segments co-evolve under 
Euler-consistent constraints. I examine parity effects in the incidence structure.  
Although Euler’s characteristic $\chi = V - E + F$ is invariant under vertex-parity changes,
quantities such as $E$, $F$, and $S$ can vary systematically between even and odd $V$.
These parity-sensitive patterns have implications for both surface enumeration
and internal tetrahedral decomposition strategies. From this, I obtain a feasible range of structural 
properties for a closed 3D polyhedron and assess their validity using Euler's formula.

\subsection{Vertex Basis and Initial Angle Unit Calculation}

Let $V \geq 4$ be the number of vertices in the polyhedron. I define:
\begin{itemize}
\item $N$: the total number of angle units on the external surface, where {\tmstrong{one unit = 180{\textdegree}}}, i.e., the interior-angle sum of one triangle.
  
\item I assume full triangulation of the surface, where polygonal faces are split into triangles using non-intersecting diagonals between existing vertices.
\end{itemize}
Then the total number of surface angle units is:
\[
N = 2 (V - 2).
\]

\begin{remark}[Triangulated baseline and flattening intuition]
When the boundary is fully triangulated (i.e.\ $S=0$), one has
$F_{\mathrm{ub}}=2V-4$ and $E_{\mathrm{ub}}=3V-6$. Interpreting $S=\sum_f(\deg f-3)$ as the
total triangulation deficit, each unit increase of $S$ corresponds to removing
one diagonal from a triangulated refinement, decreasing both $E$ and $F$ by one.
The formal identities are stated and proved in Proposition~\ref{prop:ext-formulas}.
\end{remark}

\begin{proposition}[External structure formulas]\label{prop:ext-formulas}
For a genus-$0$ polyhedral graph in normal form with $V$ vertices, $E$ edges,
$F$ faces, and face-degree excess
\[
S:=\sum_{f}(\deg f-3),
\]
the following identities hold:
\[
E = 3V - 6 - S, \qquad F = 2V - 4 - S.
\]
\end{proposition}

\begin{proof}
Because the embedding is a 2--cell embedding of $S^2$, each edge is incident to
exactly two faces, hence $\sum_f \deg f = 2E$. By definition,
$S=\sum_f(\deg f-3)=2E-3F$. Euler's formula gives $F=E-V+2$. Substituting,
\[
S = 2E - 3(E-V+2) = 3V-6-E,
\]
so $E=3V-6-S$, and then $F=E-V+2=2V-4-S$.
\end{proof}

\begin{lemma}[Lower bound edges from trivalence]\label{lem:emin-trivalent}
Let $G$ be the $1$-skeleton of a genus--$0$ polyhedron (simple, planar, and $3$-connected) with $V\ge 4$ vertices and $E$ edges. Then
\[
E \;\ge\; \left\lceil \frac{3V}{2} \right\rceil,
\]
with equality when:
\begin{enumerate}
\item $V$ is even and all vertices are trivalent (degree $3$); or
\item $V$ is odd and all but one vertex are trivalent, the remaining vertex being $4$-valent.
\end{enumerate}
\end{lemma}

\begin{proof}
Every vertex has degree at least $3$, so by the handshake lemma
\[
2E=\sum_{v\in V}\tmop{deg}(v)\ge 3V,
\]
which yields \(E\ge \lceil 3V/2\rceil\).  

If $V$ is even and all vertices have degree $3$, then \(2E=3V\) and \(E=3V/2\), attaining the bound.  

If $V$ is odd and all but one vertex are trivalent and the remaining vertex is $4$-valent, then \(2E=3(V-1)+4=3V+1\), so \(E=(3V+1)/2\), again attaining the bound.  

In the symbolic flatness framework, both equality cases correspond to $S=S_{\mathrm{ub}}$ and $F=F_{\mathrm{lb}}$ (see Corollary~\ref{cor:S_upper_bound-F_lower_bound}).
\end{proof}

\begin{corollary}[Parity restriction on the minimal edge count]\label{cor:oddV-ceiling}
Let $G$ be the $1$--skeleton of a genus--$0$ polyhedron in normal form (simple, planar, $3$--connected) with $V$ vertices and $E$ edges.  
If $V$ is odd, then $E\ne \frac{3V}{2}$ by integrality. Hence
\[
E_{\mathrm{lb}}=\left\lceil \frac{3V}{2}\right\rceil=\frac{3V+1}{2}.
\]
\end{corollary}

\begin{remark}
For odd $V$, the degree multiset achieving $E_{\mathrm{lb}}$ necessarily has one $4$--valent vertex and all others trivalent (by the handshake lemma). Realization for all odd $V$ is a separate existence question; examples exist for many $V$ (e.g., small explicit polyhedral graphs), but a full classification is beyond our scope.
\end{remark}

\begin{corollary}[Relating maximal flatness and minimal faces]\label{cor:S_upper_bound-F_lower_bound}
For $V\ge 4$, let $E_{\mathrm{lb}}=\lceil 3V/2\rceil$. Then with the formulas from Proposition~\ref{prop:ext-formulas}:
\[
S_{\mathrm{ub}}=3V-6-E_{\mathrm{lb}}
=\begin{cases}
\frac{3V}{2}-6, & \text{$V$ even},\\[2pt]
\frac{3V}{2}-\tfrac{13}{2}, & \text{$V$ odd},
\end{cases}
\qquad
F_{\mathrm{lb}}=2 - V + E_{\mathrm{lb}}
=\begin{cases}
\frac{V}{2}+2, & \text{$V$ even},\\[2pt]
\frac{V}{2}+\tfrac{5}{2}, & \text{$V$ odd}.
\end{cases}
\]
Moreover, as $S$ increases by $1$ (flattening step), both $E$ and $F$ decrease by $1$, and for
\(
k=0,1,\dots,S_{\mathrm{ub}}
\)
one has
\(
E=3V-6-k,\; F=2V-4-k.
\)
\end{corollary}

\begin{remark}
\textbf{On parity and the ``half'' constants.}
The halves (e.g., $S_{\mathrm{ub}}=\tfrac{3V}{2}-\tfrac{13}{2}$ for odd $V$) come from the integer lower bound $E_{\mathrm{lb}}=\lceil 3V/2\rceil$. Writing
\[
S_{\mathrm{ub}}=3V-6-E_{\mathrm{lb}}
\]
gives $S_{\mathrm{ub}}=\tfrac{3V}{2}-6$ when $V$ is even and $S_{\mathrm{ub}}=\tfrac{3V}{2}-\tfrac{13}{2}$ when $V$ is odd; these evaluate to integers at the respective parities. The corresponding formulas for $F_{\mathrm{lb}}=2 - V + E_{\mathrm{lb}}$ have the same source.
\end{remark}

\begin{proof}
Substitute \(E_{\mathrm{lb}}\) into the formulas in Proposition~\ref{prop:ext-formulas}.
\end{proof}

\subsection{Triangulation Sequences and the Invariance of Euler's Characteristic}
This proposition is included only to formalize the triangulated baseline and
the ``flattening ladder'' interpretation; the definitive identities are
Proposition~\ref{prop:ext-formulas} and Corollary~\ref{cor:S_upper_bound-F_lower_bound}.
\begin{proposition}[Triangulated extrema and invariance of $\chi$]\label{prop:triangulated-extrema}
For a genus-zero polyhedral surface composed entirely of triangles,
\[
F_{\mathrm{ub}}=2(V-2),\qquad E_{\mathrm{ub}}=3(V-2).
\]
If $S$ denotes the number of coplanarity merges of adjacent triangles (each merge decreases both $F$ and $E$ by $1$ while $V$ is fixed), then for all $S$ in the admissible range,
\[
F(S)=F_{\mathrm{ub}}-S,\qquad E(S)=E_{\mathrm{ub}}-S,\qquad V-E(S)+F(S)=2.
\]
\end{proposition}

\begin{proof}
Each triangle has three edges and each edge is shared by two triangles, so $2E_{\mathrm{ub}}=3F_{\mathrm{ub}}$. Substituting $E_{\mathrm{ub}}=\tfrac{3}{2}F_{\mathrm{ub}}$
into Euler’s formula $V-E_{\mathrm{ub}}+F_{\mathrm{ub}}=2$ yields $F_{\mathrm{ub}}=2(V-2)$ and $E_{\mathrm{ub}}=3(V-2)$. A coplanarity merge of two adjacent triangles removes
one surface edge and one triangle, so $F\mapsto F-1$ and $E\mapsto E-1$ with $V$ unchanged. Hence $V-E+F$ is invariant and equal to $2$ for all admissible $S$.
\end{proof}

\paragraph{Equivalent parameterization for fixed $V$ from the lower extremum.}
Equivalently, indexing the flattening ladder from the lower extremal configuration
$(S_{\mathrm{ub}},E_{\mathrm{lb}},F_{\mathrm{lb}})$, one may write
\[
S_j = S_{\mathrm{ub}} - j, \qquad
E_j = E_{\mathrm{lb}} + j, \qquad
F_j = F_{\mathrm{lb}} + j,
\]
for $j = 0,1,\dots,S_{\mathrm{ub}}$.
This parameterization enumerates all admissible external configurations consistent
with the flattening rule and Euler invariance.

This result shows that the maximal triangulated case (\(S=0\)) fixes $F$ and $E$ uniquely in terms of $V$, and that the Euler characteristic is preserved under
a sequence of coplanarity merges parameterized by $S$. In our framework, $S$ also functions as a discrete measure of flatness and appears in parity-sensitive
bounds in this section, and will reappear in linking external and internal structures (\S5). Appendix~A lists worked examples for $V=4$ to $V=7$ illustrating
these sequences and confirming $V-E+F=2$ in all cases.

\begin{heuristicproposition}[Flatness values in the admissible range can be realized]\label{prop:realizability}
Fix $V\ge 4$. For every integer $S$ with $0\le S\le \frac{3V}{2}-6$ when $V$ is even and $0\le S\le \frac{3V}{2}-\tfrac{13}{2}$ when $V$ is odd, one can construct a genus-$0$ polyhedron in normal form (equivalently, a simple $3$–connected planar graph) with
\[
E=3V-6-S,\qquad F=2V-4-S.
\]
In particular, every integer flatness value in the admissible range is realized.
\end{heuristicproposition}

\begin{remark}
Start from a maximally planar (triangulated) $3$--connected planar graph on
$V$ vertices, with $S=0$, $E=3V-6$, and $F=2V-4$.
Deleting a suitable shared edge of two adjacent boundary triangles decreases
both $E$ and $F$ by $1$ while keeping $V$ fixed, corresponding to an increase
of $S$ by $1$.
In stacked triangulations, such deletions provide a natural constructional
route through part of the admissible range.
\end{remark}

\subsection{External Angle Unit Bounds}

From the flatness-based identities in Proposition~\ref{prop:ext-formulas}, 
and holding $V$ fixed, the external angle unit range can be expressed as:
\[
M_{\mathrm{ub}} = 10 (V - 2), \qquad 
M_{\mathrm{lb}} = \begin{cases}
4V + 6, & \text{$V$ odd},\\
4V + 4, & \text{$V$ even}.
\end{cases}
\]

\subsection{Counting Distinct Combinatorial Configurations}

By combining Proposition~\ref{prop:ext-formulas} and Corollary~\ref{cor:S_upper_bound-F_lower_bound},  
the total number of distinct $(F,E)$ configurations attainable by varying $S$ over 
its admissible range is:
\[
\text{Combinations} = S_{\mathrm{ub}} + 1 =
\begin{cases}
\frac{3}{2} V - \tfrac{11}{2}, & \text{$V$ odd},\\
\frac{3}{2} V - 5, & \text{$V$ even}.
\end{cases}
\]

\begin{heuristic}[Realizability across intermediate flatness]\label{heur:external-realizability}
Proposition~\ref{prop:ext-formulas} and Lemma~\ref{lem:emin-trivalent} imply the
admissible parity-sensitive ranges for $E$, $F$, and $S$ in
Corollary~\ref{cor:S_upper_bound-F_lower_bound}. What remains heuristic in this work is the claim
that \emph{every} integer $S$ in the admissible range is realizable by a simple
$3$--connected planar graph in normal form; see
Heuristic Proposition~\ref{prop:realizability}.
\end{heuristic}

\paragraph*{Angle units.}
One ``angle unit'' equals the interior-angle sum of one triangle, i.e., $180^\circ$.
On a fully triangulated boundary, each triangle contributes exactly one unit and
\[
N = F_{\triangle} = F + S.
\]
Under an interior decomposition into $T$ tetrahedra with $N_i$ internal gluing triangles, 
counting triangle units gives
\[
N = 4T - 2N_i,
\]
since each tetrahedron contributes four triangle units and every internal gluing triangle 
removes two boundary triangle units by pairing. These two expressions for $N$ are consistent 
with Proposition~\ref{prop:boundary-from-interior} and Proposition~\ref{prop:ext-formulas}: 
indeed $F+S=4T-2N_i$. This connection between external angle counts and internal tetrahedral 
composition will be used in \S5 to link internal volumetric parameters directly to external 
face and edge counts.

\section{Internal Decomposition and Tetrahedral Estimation}

This section shifts from surface validation to internal volumetric structure.
We record exact incidence identities governing tetrahedral decompositions of
genus--0 polyhedra in normal form, then study how restricted construction rules and surface structure constrain realizability.

\paragraph*{Scope.}
This section assumes normal form (Definition~\ref{def:normal-form}).
The empirical analyses additionally use the restricted decomposition procedure
of \S3.6, in which tetrahedra are removed progressively from the boundary inward
and at most one new interior segment is introduced at each step.

\subsection{Determining admissible tetrahedral ranges from polygonal surface data}
\label{sec:tetrahedral-range-method}

For a closed genus--0 polyhedral surface, the relevant boundary and interior
quantities are:

\begin{itemize}
    \item $S=\sum_f(\deg f-3)$, the boundary face-degree excess;
    \item $N=F+S=2V-4$, the number of boundary triangles after triangulation;
    \item $E_\triangle=E+S=3V-6$, the number of boundary edges after triangulation;
    \item $T$, $N_i$, and $S_i$, denoting tetrahedra, internal triangular
    gluing faces, and interior triangulation segments, respectively.
\end{itemize}

Triangle incidence and the Euler characteristic of the tetrahedralized
$3$--ball give
\[
4T=N+2N_i,
\qquad
T-N_i+S_i=1
\]
in normal form. Hence
\[
N_i=2T-V+2,
\qquad
S_i=T-V+3,
\]
and therefore
\[
\boxed{T\ge V-3}.
\]

For the restricted decomposition procedure, at most one new interior segment
is introduced per layer, with the first and last layers unable to host such a
segment. The construction-based range considered here is therefore
\[
0\le S_i\le V-5.
\]
Together with the preceding identities, this gives
\[
T_{\min}=V-3,\qquad
T_{\max}=2(V-4),
\]
and
\[
N_{i,\min}=V-4,\qquad
N_{i,\max}=3(V-4)-2,\qquad
S_{i,\max}=V-5.
\]

\subsection{Exact interior incidence identities (normal form)}

\begin{proposition}[Exact 3D Euler link]\label{prop:exact-euler}
Let $T$ be the number of tetrahedra in a genus--0 tetrahedralization,
$N_i$ the number of internal gluing triangles,
$E_i$ the number of interior edges,
and $V_i$ the number of interior vertices.
Then
\[
\boxed{T - N_i + E_i - V_i = 1.}
\]
\end{proposition}

\begin{proposition}[Angle units and internal gluings]\label{prop:Ni-exact}
Let a genus--$0$ polyhedron in normal form be tetrahedralized into $T$ tetrahedra
with $N_i$ internal gluing triangles.
After triangulating all boundary faces, the number of boundary triangles is
$N=2V-4$.
Counting triangle units gives
\[
4T - 2N_i = N \quad \Longrightarrow \quad
\boxed{N_i = 2T - V + 2.}
\]
\end{proposition}

\begin{proposition}[Normal-form specialization]\label{prop:TNisiexact}
In normal form, $V_i \equiv 0$ and every interior edge is an internal
triangulation segment; define $S_i := E_i$.
Then
\[
\boxed{T - N_i + S_i = 1,}
\]
and combining with Proposition~\ref{prop:Ni-exact} yields
\[
\boxed{T = V - 3 + S_i.}
\]
\end{proposition}

These identities are exact and independent of any construction method.

\subsection{General solution family and incidence bounds}

For fixed boundary vertex count $V$, Propositions
\ref{prop:Ni-exact}--\ref{prop:TNisiexact} place every normal-form
tetrahedralization on the integer family
\[
N_i(T)=2T-V+2,\qquad S_i(T)=T-V+3.
\]
These equations give necessary incidence relations.

\subsection{The restricted decomposition procedure ladder family}

Within the restricted decomposition procedure, a transparent subfamily is the “ladder” construction where a polyhedron is decomposed into tetrahedra by the identification of internal triangular gluing faces one by one. For $k=0,\ldots,V-5$,
\[
T_k = V-3+k,\quad N_{i,k} = V-4+2k,\quad S_{i,k}=k,
\]
and $T_k-N_{i,k}+S_{i,k}=1$ holds exactly.

For fixed $V$, the ladder construction realizes exactly $V-4$
distinct admissible triples $(T,N_i,S_i)$ under the restricted decomposition procedure.

\begin{heuristic}[Internal ranges and configuration ladder]\label{heur:internal-ranges}
The ranges for $(T_{\min},T_{\max})$, $(N_{i,\min},N_{i,\max})$, and $S_{i,\max}$, together with the ladder
\[
T_k=V-3+k,\quad N_{i,k}=V-4+2k,\quad S_{i,k}=k\quad (0\le k\le V-5),
\]
match all worked examples and constructions tested to date, but are presently justified empirically rather than by a complete classification.
\end{heuristic}

\subsection{Empirical truncation at fixed $V$}

Although the ladder exhausts the algebraic range generated by the restricted
decomposition procedure, realizability at
fixed $V$ depends on the external surface structure $(E,F,S)$. Table~\ref{tab:V8} (for $V=8$) exhibits a sharp compression: as surface
flatness $S$ increases, progressively fewer ladder rungs are realized.
For example, when $S=6$ (the cube), the restricted successive-peeling procedure terminates at $T=5$. The familiar six-tetrahedron decomposition belongs to a different tetrahedralization pathway involving preliminary
repartitioning into two triangular prisms and is therefore outside the restricted procedure. 

\begin{table}[H]
\centering
\caption{Currently recorded $(T,N_i,S_i)$ configurations at $V=8$ under the restricted decomposition procedure,
grouped by external flatness $S$. Here $E$ denotes the number of boundary edges
prior to triangulating non-triangular faces. For fixed $V=8$, the triangulated
boundary satisfies $E+S=3V-6=18$.}
\label{tab:V8}
\setlength{\tabcolsep}{7pt}
\renewcommand{\arraystretch}{0.9}
\begin{tabular}{ccccc}
\toprule
$E$ & $S$ & $T$ & $N_i$ & $S_i$ \\
\midrule
18 & 0 & 5 & 4  & 0\\
18 & 0 & 6 & 6  & 1\\
18 & 0 & 7 & 8  & 2\\
18 & 0 & 8 & 10 & 3\\
\midrule
17 & 1 & 5 & 4 & 0\\
17 & 1 & 6 & 6 & 1\\
17 & 1 & 7 & 8 & 2\\
\midrule
16 & 2 & 5 & 4 & 0\\
16 & 2 & 6 & 6 & 1\\
\midrule
15 & 3 & 5 & 4 & 0\\
\midrule
14 & 4 & 5 & 4 & 0\\
\midrule
13 & 5 & 5 & 4 & 0\\
\midrule
12 & 6 & 5 & 4 & 0\\
\bottomrule
\end{tabular}
\end{table}
\begin{remark}[Status of the empirical table]
The empirical basis of Table~\ref{tab:V8} is being revalidated under the
clarified successive-peeling introduced in Section~3.6.
In particular, the $S=3$, $S=4$, and $S=5$ entries have not yet been
rechecked against the requirement that previously established internal
structure is retained and not repartitioned or retriangulated. These rows
should therefore be regarded as provisional in the present version.
\end{remark}

From these configurations, the following empirical tendencies are provisionally observed for eight-vertex polyhedra:

\begin{itemize}
  \item When $S = 0$: all $V-4$ ladder configurations are realized.
  \item When $S = 1$: the top ladder rung $(S_i=S_{i,\max})$ is not realized.
  \item When $S = 2$: the two lowest observed ladder rungs are
$(S_i=0,1)$.
  \item When $3 \le S \le 6$: only the lowest rung $(S_i=0)$ is realized under the
restricted successive-peeling procedure.
\end{itemize}

\subsection{Flexibility and Inflexibility of Internal Arrangements at Fixed $V$}

The preceding sections establish admissible ranges for the number of
tetrahedra $T$ and internal gluing triangles $N_i$ for a polyhedron with
fixed boundary vertex count $V$ under the restricted decomposition procedure.
However, these bounds alone do not distinguish between configurations that
admit multiple internal realizations and those that are structurally rigid.
In this section, we introduce a qualitative distinction between
\emph{inflexible} and \emph{flexible} internal arrangements, based on whether
the surface triangulation budget $S$ can vary while $V$ and $T$ are held
fixed.

\subsubsection{Definition: $S$-flexibility}

Fix a boundary vertex count $V$ and a tetrahedral count $T$.
We say that a configuration is \emph{$S$-inflexible} if all valid
restricted decomposition procedure realizations with these fixed values of $(V,T)$ admit the same
value of $S$.
Conversely, a configuration is \emph{$S$-flexible} if there exist two valid
realizations with the same $(V,T)$ but different values of $S$.

Intuitively, $S$-inflexibility corresponds to a boundary and interior
structure that forces a unique value of the surface triangulation budget,
while $S$-flexibility allows that budget to vary without changing the
global counts $V$ or $T$.

\subsubsection{Inflexible class: bipyramidal configurations}

The canonical examples of $S$-inflexible arrangements are
\emph{bipyramidal configurations}.
These consist of two apex vertices (one ``above'' and one ``below'')
connected to all vertices of a central polygonal cycle.
All faces are triangular, and the interior tetrahedra fan out from the
central polygon in a forced manner.

For such configurations:
\begin{itemize}
  \item the induced tetrahedral adjacency structure is fixed,
  \item no local rearrangement of internal gluing surfaces is possible
        without changing either $V$ or $T$,
  \item the surface triangulation budget $S$ is uniquely determined
        (typically $S=0$).
\end{itemize}

Among the manually constructed $(V,T)=(8,8)$ examples, the only observed realizable configuration is the hexagonal bipyramid.
In this case, no alternative triangulation of the surface exists that is
compatible with a valid tetrahedral decomposition, and the configuration is
strictly $S$-inflexible.

\subsubsection{First deviation: single-vertex departure from bipyramidity}

Flexibility first appears when the bipyramidal structure is minimally
perturbed.
Specifically, if one vertex leaves the equatorial polygon and is reassigned
to either the upper or lower apex side, the central polygon is reduced by
one vertex.
One side of the polyhedron remains a pyramid, while the opposite side now
contains a single quadrilateral face.

The corresponding four-vertex boundary region can occur in two distinct
surface incidence configurations:
\begin{itemize}
  \item as two triangular faces separated by a boundary edge, giving $S=0$;
  \item as a single quadrilateral face, giving $S=1$.
\end{itemize}

Crucially, both realizations share the same values of $(V,T)$.
Thus, this class represents a \emph{minimally flexible} arrangement, with
exactly one discrete degree of freedom in $S$.
The interior accommodates one fewer tetrahedron than the corresponding pure
bipyramid, but remains otherwise tightly constrained.

\subsubsection{Higher deviations: increasing flexibility}

Additional flexibility arises when further vertices leave the equatorial
polygon.
For example:
\begin{itemize}
  \item two vertices may be reassigned to the same side, producing a
        pentagonal face;
  \item or one vertex may move upward while another moves downward,
        producing a quadrilateral face on each side.
\end{itemize}

Each non-triangular face introduces different ways of dividing the surface into triangles.
As a result, the set of admissible values of $S$ grows, while $V$ and $T$
remain fixed.
These configurations are genuinely $S$-flexible, exhibiting multiple
realizable internal arrangements and a nontrivial range of surface
triangulation budgets.

\subsection{Interpretation}

This progression suggests a natural stratification of configurations at
fixed $V$:
\begin{center}
\emph{bipyramidal (rigid)} $\;\rightarrow\;$
\emph{single-defect (minimally flexible)} $\;\rightarrow\;$
\emph{multi-defect (flexible)}.
\end{center}

In this view, bipyramidal configurations form the rigid core of the
admissible space, while flexibility emerges precisely as vertices are
removed from the equatorial cycle.
Each such departure introduces a higher-degree face and hence more ways of dividing the surface into triangles, reflected in the admissible range of surface flatness $S$.

This distinction clarifies why the admissible restricted decomposition procedure bounds alone do not
determine $T$ or $S$ uniquely: realizability within those bounds depends
sensitively on the internal combinatorial structure, not solely on boundary
counts.
\paragraph{Remark (Incidence-Gated Transitions, Bipyramidal Rigidity, and Boundary Flexibility).}
Although the analysis throughout this work is entirely symbolic, it is often helpful to imagine vertices as concrete points in space and tetrahedra as occupying volume. This geometric picture is not essential, but it helps explain why certain discrete effects appear abruptly rather than gradually.

A key point is that the equatorial ``belt'' in a bipyramidal configuration need not be coplanar. Coplanarity is a red herring: what actually governs admissibility is the pattern of \emph{external boundary edges} and \emph{internal gluing segments}. As long as these incidence relations are preserved, the belt vertices can ripple or kink freely in space without any effect on the tetrahedral count or on other boundary quantities (edges and faces). Geometry supplies continuous freedom within each combinatorial state.

From this perspective, bipyramidal configurations occupy a special extremal position. Their rigidity is geometrically obvious: two apices constrain a belt of vertices so tightly that there is essentially no slack. However, this rigidity is largely invisible at the level of surface incidence counts alone. In informal terms, bipyramids are often mistaken for ``flexible'' objects combinatorially, even though they are maximally constrained volumetrically. This is why, if one insists on geometric metaphors for rigidity, the usual description of inflexible objects as ``squares'' is misleading: squares are flexible in triangulated settings and signal at least a degree of boundary flexibility, whereas bipyramids are not.

Departures from bipyramidity occur through discrete reclassifications of vertex roles rather than smooth geometric limits. The primitive defect is not an apex becoming coplanar with the belt, but the reassignment of a \emph{belt vertex}. When a single belt vertex ceases to attach symmetrically to both apices, remaining connected to only one apex while acquiring additional lateral connections, the belt shortens by one, an internal gluing segment is lost, and exactly one tetrahedron disappears. This produces the first unit of variability: two triangular faces may unite to form a boundary quadrilateral ($S=1$), the tetrahedral count drops by one, and the boundary geometry remains otherwise flexible, but now within a new combinatorial state.

A stronger transition occurs if deformation proceeds far enough that an \emph{apex--apex edge} appears. In this regime the belt is no longer merely reclassified locally; it is shortened more substantially, and multiple tetrahedra are lost simultaneously. In an eight-vertex system, for example, this transition produces a loss of two or three tetrahedra (depending on how the polygon is decomposed) relative to the bipyramidal maximum. These changes are discrete and unavoidable: intermediate geometric states do not register combinatorially because the system skips them by reorganizing its incidence structure.

This behavior illustrates a general asymmetry between boundary and interior. Continuous geometric variation of the boundary can induce discrete changes in interior capacity, while fixing the combinatorial structure of the tetrahedral decomposition leaves substantial freedom in the realization of the boundary. The resulting ``quantized'' behavior of tetrahedral and boundary counts is not a physical effect, but a consequence of incidence-gated admissibility: counts remain constant under arbitrary rippling until a gate is crossed, at which point they jump.

\subsection{Inferring external structure from internal decomposition}
\label{sec:infer-external-from-internal}

This subsection relates the internal tetrahedral decomposition of a genus--$0$
polyhedron to the combinatorics of its boundary surface. We recall the exact
topological identities for tetrahedralizations of a $3$--ball and their
normal--form specialization (already established in
Propositions~\ref{prop:exact-euler}--\ref{prop:TNisiexact}), then record a
useful incidence-based rule-of-thumb for reconstructing $(V,E,F)$ from interior
decomposition data.

\subsubsection{Exact Euler link and normal-form specialization}

For any tetrahedral complex filling a $3$--ball, the exact identity
(Proposition~\ref{prop:exact-euler})
\[
T - N_i + E_i - V_i = 1
\]
always holds. In the normal form adopted throughout this manuscript,
$V_i\equiv 0$ and every interior edge is an internal triangulation segment, so
$E_i=S_i$ and hence (Proposition~\ref{prop:TNisiexact})
\[
\boxed{\,T - N_i + S_i = 1\,}.
\]
Moreover, after triangulating all boundary faces, triangle incidence gives
(Proposition~\ref{prop:Ni-exact})
\[
\boxed{\,N_i = 2T - V + 2\,}.
\]
Together these identities link the interior parameters $(T,N_i,S_i)$ to the
boundary vertex count $V$ in a purely symbolic way.

\subsubsection{Boundary counts from interior decomposition}

The next statement is \emph{not} a general topological law; it is an incidence
calculation that assumes:
\begin{enumerate}
\item no interior Steiner vertices ($V_i=0$),
\item all internal adjacencies occur along triangular faces,
\item $S=\sum_f(\deg f-3)$ is the face-degree excess (equivalently the number of boundary triangulation diagonals),
\item $S_i$ counts internal triangulation segments completing gluing surfaces.
\end{enumerate}

\begin{proposition}[Boundary counts from interior decomposition]
\label{prop:boundary-from-interior}
Assume a genus--$0$ decomposition into $T$ tetrahedra without interior Steiner
vertices, where internal adjacencies occur only along triangular faces; let
$N_i$ be the number of internal gluing triangles, $S$ the number of coplanarity
merges on the boundary, and let $S_i$ denote the number of interior edges.
Then the boundary satisfies
\[
E=6T-3N_i-S,\qquad
F=4T-2N_i-S,\qquad
V=4T-3N_i+2S_i.
\]
\end{proposition}

\begin{proof}[Justification]
For a fully triangulated boundary, $F_\triangle=4T-2N_i$ and
$E_\triangle=\tfrac{3}{2}F_\triangle=6T-3N_i$ by triangle incidence.
Coplanarity merges reduce both $E_\triangle$ and $F_\triangle$ by $S$, yielding
the displayed formulas for $E$ and $F$. The $V$--formula is algebraically
equivalent to the exact identities $N_i=2T-V+2$ and $T-N_i+S_i=1$
(Propositions~\ref{prop:Ni-exact}--\ref{prop:TNisiexact}), hence consistent.
\end{proof}

\paragraph{Sanity check.}
In normal form, $T-N_i+S_i=1$ (Proposition~\ref{prop:TNisiexact}), so the above
formulas automatically satisfy $V-E+F=2$ after substitution.

\subsubsection{Boundary--interior Euler correspondence}

Under the normal-form assumptions, the boundary Euler relation and the exact
interior incidence identity are linked by
\[
V-E+F = 2\,(T-N_i+S_i).
\]

\begin{proposition}[Boundary--interior Euler correspondence]
\label{prop:euler-link}
For a genus--$0$ tetrahedralization in normal form,
\[
V-E+F = 2\,(T-N_i+S_i).
\]
\end{proposition}

\begin{proof}
The boundary satisfies
\[
V-E+F=2.
\]
In normal form, $V_i=0$ and $S_i=E_i$, so
Proposition~\ref{prop:TNisiexact} gives
\[
T-N_i+S_i=1.
\]
Therefore
\[
V-E+F=2(T-N_i+S_i).
\]
\end{proof}

Appendix~F compiles all identities and bounds derived in Sections~4 and~5.
\clearpage
\section{A Combinatorial Method for Testing 3D Polygonal Enclosure and
Polyhedral Feasibility}

Section~6 formalizes a combinatorial method to assess polygon lists for
enclosure feasibility, including a ``flatness threshold'' test. The method
operates on a list of polygon types and their frequencies (e.g., 12 pentagons,
20 hexagons), without requiring coordinates or spatial modeling. It proceeds
in several steps:

\paragraph{\textbf{Step 1: Face frequencies.}}
Let $F_k$ denote the number of boundary faces of degree $k$.

\

\paragraph{\tmtextbf{Step 2: Total Internal Angle Units (N)}:} Each $k$-gon
contributes $(k - 2)$ angle units of 180{\textdegree}. Compute total:
\[ N = \sum_{k = 3}^n (k - 2) F_k \]

\
 
\paragraph{\tmtextbf{Step 3: Vertex Count Estimate (V)}}  
For any genus-$0$ polyhedron in normal form (simple $3$–connected boundary; no edge bisections or interior boundary vertices), the vertex count is governed by a direct formula in terms of the face–degree counts.

\begin{proposition}\label{prop:V-from-N}
Let $F_k$ denote the number of faces of degree $k$, and define
\[
N \;=\; \sum_{k \ge 3} (k-2)\,F_k.
\]
Then for any finite simple graph with a 2--cell embedding in $S^2$ (genus $0$),
\[
\boxed{\;V \;=\; \frac{N}{2} + 2\;}
\]
\end{proposition}

\begin{proof}
From the face–edge incidence relation,
\[
2E \;=\; \sum_{k \ge 3} k\,F_k
  \;=\; \sum_{k \ge 3} (k-2)\,F_k \;+\; 2\sum_{k \ge 3} F_k
  \;=\; N \;+\; 2F.
\]
By Euler’s formula $V - E + F = 2$, substituting $E = (N + 2F)/2$ gives
\[
V - \frac{N + 2F}{2} + F = 2
\quad\Rightarrow\quad
V - \frac{N}{2} = 2,
\]
and hence $V = \frac{N}{2} + 2$.
\end{proof}

\begin{corollary}[Integer feasibility check]\label{cor:int-check}
If $N$ is odd, then $V$ is non-integer and no valid genus-0 polyhedron in normal form with such a face–degree distribution can exist.
\end{corollary}

\medskip
\noindent\emph{Remark (origin of the formula).}  
An earlier derivation produced the more elaborate rational form  
\[
V = \frac{2 M_{\mathrm{total}} - N \,(2 S_{\mathrm{total}} + 10)}{M_{\mathrm{total}} - 5 N},
\]
with 
\[
M = \sum_{k\ge 3}(k+2)F_k, 
\quad S = \sum_{k\ge 3}(k-3)F_k,
\]
arising from edge--face incidence counts and the extremal relation $E=3V-6-S$.  
Under the genus--$0$, normal--form constraints, one has $M-5N = -4S$, which causes the above expression to collapse identically to $V=\tfrac{N}{2}+2$ for all values of $S$.  
This shows that the vertex count is \emph{independent} of flatness in this setting.

\

\indent\emph{Historical derivation}

Let $F_k$ be the number of $k$--gonal faces, and define
\[
N=\sum_{k\ge 3}(k-2)F_k,\quad
M=\sum_{k\ge 3}(k+2)F_k,\quad
S=\sum_{k\ge 3}(k-3)F_k,\quad
F=\sum_{k\ge 3}F_k.
\]
From
\[
2E=\sum_{k\ge 3}kF_k,\qquad 2E=3F+S,
\]
we find 
\[
M = (2E)+2F,\quad N = (2E)-2F,
\]
and hence
\[
E=\frac{M+N}{4},\quad F=\frac{M-N}{4},\quad
M-5N = -4(2E-3F) = -4S.
\]
With $E=3V-6-S$ (Proposition~\ref{prop:ext-formulas}), we have
\[
\frac{M+N}{4} = 3V - 6 - S
\quad\Longrightarrow\quad
V=\frac{M+N+24+4S}{12}.
\]
Substituting $M=5N-4S$ yields
\[
V = \frac{N}{2} + 2.
\]
Thus the simple form follows unconditionally from the general counts.

\

\paragraph{\tmtextbf{Step 4: Face-degree excess (flatness) $S$:}}
\[
S := \sum_{k\ge 3} (k-3)\,F_k.
\]

\

\paragraph{\tmtextbf{Step 5: Total Face Count (F)}:}

\[ F = \sum_{k = 3}^n F_k \]

\

\paragraph{\tmtextbf{Step 6: Edge count (E)}:}
\[
E = 3V - 6 - S.
\]
Equivalently, after triangulation one has $F_\triangle = F+S$ and $E_\triangle = E+S$,
and the triangulated incidence identity
\[
2E_\triangle = 3F_\triangle
\quad\text{ i.e. }\quad
2(E+S)=3(F+S)
\]
holds.

\

\paragraph{\tmtextbf{Step 7: Euler's Consistency Check}:}

\[ V - E + F = 2 \]
Because $V$ and $E$ above are derived from the same Euler--incidence relations,
this serves as an internal consistency check rather than an independent
feasibility filter.

\

\paragraph{\tmtextbf{Step 8: Flatness Constraint}:}

\[ S_{\mathrm{lb}} = 0 \]
\[ S_{\mathrm{ub}} = \left\{ \begin{array}{ll}
     \frac{3}{2} V - \tfrac{13}{2} & \text{if } V \text{is odd}\\
     & \\
     \frac{3}{2} V - 6 & \text{if } V \text{is even}
   \end{array} \right. \]
\[ S_{\mathrm{diff}} = S_{\mathrm{ub}} - S\]
If $S > S_{\mathrm{ub}}(V)$, the face-degree excess exceeds the parity/degree upper bound
implied by $\delta(G)\ge 3$, so no such genus-zero polyhedral graph can exist.
Loosely, one may say the proposed surface is \emph{``too flat (combinatorially) to enclose''}.

\paragraph{\textbf{Step 9: Local face-degree feasibility}}\footnote{
I had missed this local condition when putting the framework together in my
teenage years. It was brought to my attention during preparation of this
version through a ChatGPT-assisted review of the screening procedure.
}

Let
\[
k_{\max}=\max\{k:F_k>0\}.
\]
For a simple $3$--connected planar polyhedral graph, the dual graph is also
simple. Since the dual has $F$ vertices, every dual vertex has degree at most
$F-1$. Therefore every face must satisfy
\[
\deg(f)\le F-1,
\]
or equivalently,
\[
k_{\max}\le F-1.
\]

If $k_{\max}>F-1$, the candidate is rejected even if the global Euler and
flatness checks are satisfied. For example, two hexagons and four triangles
give $F=6$ and $k_{\max}=6$, so the global checks pass but the local condition
fails because $6>5$.

\paragraph{Practical examples.}
Worked applications of the symbolic screening procedure to several polygonal face configurations are provided in Appendix~D.
\clearpage
\section{Symmetry Between Face- and Vertex-Type Combinatorics}

This section investigates the symmetry between face types and
vertex valencies, drawing attention to a duality often noted but rarely
operationalized in enclosure logic. In analyzing polyhedra from a combinatorial
standpoint, we observe a notable parallelism between the types of faces and the
types of vertices.

\subsection{Face Types and Vertex Types}

Each face of a polyhedron is defined by the number of its edges (or
equivalently, the number of vertices around it), while each vertex is defined
by the number of edges meeting at that point. These quantities are known as:
\begin{itemize}
  \item \textbf{Face degrees}: the number of edges surrounding a face,
  \item \textbf{Vertex valencies}: the number of edges incident to a vertex.
\end{itemize}

For a polyhedron with $V$ vertices, $E$ edges, and $F$ faces, we always have
the exact incidence identity
\[
  \sum_{i = 1}^F d_i \;=\; 2 E \;=\; \sum_{j = 1}^V v_j,
\]
where $d_i$ is the degree of the $i$-th face and $v_j$ is the valency of the
$j$-th vertex. This face--vertex symmetry is emphasized in Gr{\"u}nbaum’s
treatment of polyhedral duality \cite{gruenbaum2003}, where vertex valencies
and face degrees are shown to reflect a deeper structural reciprocity across
convex forms. It reflects a basic balance: each edge is shared by two faces and
connects two vertices, hence contributes twice to the totals.

\subsection{Observed Symmetry and Interpretation}

From empirical examination and combinatorial modeling, I notice that:
\begin{itemize}
  \item The range of face types (triangle, quadrilateral, pentagon, etc.)
  mirrors the range of vertex types (valency 3, 4, 5, etc.).
  \item The minimum possible face or vertex type is generally 3 (a triangle or
  trivalent vertex), which is necessary for a polyhedron to be enclosed.
  \item The maximum number of edges in a face or incident edges at a vertex is
  bounded by the number of vertices, typically less than $V$.
\end{itemize}

This observation does not suggest a strict one-to-one mapping but highlights a
structural analogy:
\begin{itemize}
  \item When enumerating or constructing valid polyhedra, both sets of types
  must satisfy the same total:
  \[ \text{Total face degrees} = \text{Total vertex valencies} = 2 E. \]
  \item This constrains generation algorithms and supports validation
  strategies.
  \item It provides a balanced perspective when analyzing polyhedra from the
  inside (vertices) or outside (faces).
\end{itemize}
Though rooted in standard topological rules, this symmetry has not been
prominently emphasized in the literature as a guiding structure. It suggests
an intuitive, combinatorial harmony between internal and external complexity.

\subsection{Face-Type Combinatorial Constraints}

We now narrow down the valid face-type combinations in a polyhedron vertex count $V$, edge count $E$, face count $F$, and external flatness $S$.

\paragraph{Face-Type Limit Rule.}
For a simple $3$--connected planar polyhedral graph, the dual graph is also
simple. A face of degree $k$ corresponds to a dual vertex of degree $k$, and
therefore
\[
k \le F-1.
\]
For example, the candidate consisting of two hexagons and four triangles has
$F=6$ and satisfies the global Euler and flatness checks, but a hexagonal face
would require a degree-$6$ vertex in a simple dual graph with only six
vertices, which is impossible.

\paragraph{Equations for Face-Type Combinations.}
For given values of $F$, $E$, and $S$, the face-type variables must satisfy:
\[
\begin{aligned}
F_3 + F_4 + F_5 + \cdots &= F
  && \text{(total number of faces)},\\
3F_3 + 4F_4 + 5F_5 + \cdots &= 2E
  && \text{(edge--face incidence)},\\
(3{-}3)F_3 + (4{-}3)F_4 + (5{-}3)F_5 + \cdots &= S
  && \text{(external triangulation segments)}.
\end{aligned}
\]
All solutions must be non-negative integers.

\paragraph{Worked Example ($V=6$, $S=2$, $E=10$, $F=6$).}
Allowed face types: triangles, quadrilaterals, and pentagons (up to 5 sides).
Valid integer solutions:
\begin{itemize}
  \item $F_3 = 4$, $F_4 = 2$, $F_5 = 0$ \quad (4 triangles, 2 quadrilaterals),
  \item $F_3 = 5$, $F_4 = 0$, $F_5 = 1$ \quad (5 triangles, 1 pentagon).
\end{itemize}
Both satisfy:
\[
\begin{aligned}
  F_3 + F_4 + F_5 & = 6 \\
  3 F_3 + 4 F_4 + 5 F_5 & = 20 \\
  F_4 + 2 F_5 & = 2.
\end{aligned}
\]

\subsection{Vertex-Type Combinatorial Constraints}

We now narrow down the valid vertex-type combinations in a polyhedron with total vertex count $V$ and edge count $E$.

\paragraph{Vertex Type System of Equations.}
Every valid vertex-type configuration must satisfy:
\[
\begin{aligned}
     V_3 + V_4 + V_5 + \ldots & = V & \quad & \text{(vertex count)} \\
     3 V_3 + 4 V_4 + 5 V_5 + \ldots & = 2 E & \quad & \text{(edge incidence)} \\
     V_3 + 2 V_4 + 3 V_5 + \ldots & = 2E - 2V
     & \quad & \text{(since $\sum_j (\tmop{deg} v_j-2) = 2E - 2V$).}
\end{aligned}
\]
All variables must be non-negative integers, and the maximum degree allowed is
$V-1$.

\paragraph{Worked Example ($V=6$, $E=10$, $S=2$, $T=3$).}
Solutions include:
\begin{itemize}
  \item $V_3 = 4$, $V_4 = 2$, $V_5 = 0$,
  \item $V_3 = 5$, $V_4 = 0$, $V_5 = 1$.
\end{itemize}
Both satisfy all constraints.

\subsection{Limitations and Realizability Filters}

This illustrates a key limitation: satisfying combinatorial constraints does
not guarantee embeddability in three-dimensional space. As Gr{\"u}nbaum notes,
only graphs that are planar and 3-connected can be realized as convex 3D
polyhedra, a result due to Steinitz \cite{gruenbaum2003}.  

To reduce false positives, one may introduce heuristic filters:
\begin{itemize}
  \item Bounding the number of high-valency vertices relative to the number of
  tetrahedra,
  \item Angle feasibility checks,
  \item Limits on repeated vertex incidence across tetrahedra.
\end{itemize}
Such filters are not definitive, but they help identify configurations that
are combinatorially possible yet geometrically implausible.

\subsection{Summary}

In summary, the structural symmetry between face-type and vertex-type
combinatorics provides a unifying framework. The parallel systems of equations
show that both perspectives encode the same edge-incidence information while
imposing complementary feasibility conditions. This dual viewpoint clarifies
why some integer solutions admit realizable polyhedra while others fail,
highlighting the balance between external and internal complexity in
polyhedral structures.

\newpage
\section{On All Likely Configurations of a Polyhedron and The Limits of
Combinatorial Realizability}

Section~8 explores a possible route for exploring all likely configurations
of a polyhedron defined by a number of vertices and derived surface flatness
by utilizing \ an unrestricted partition function $p (S) .$ As shown in
earlier sections, every valid polyhedron of $V$ vertices is characterized by
two distinct combinatorial structures:
\begin{itemize}
  \item The external structure: $(F, E, M, S)$. For a polyhedron of vertices V,
  these complements have a total number of combinations$= \left\{
  \begin{array}{ll}
    \frac{3}{2} V - \tfrac{11}{2} & \text{if } V \text{is odd}\\
    \frac{3}{2} V - 5 & \text{if } V \text{is even}
  \end{array} \right.$
  
  \item The internal structure: $(T, N_i, S_i)$. For a polyhedron of
  vertices V, these complements have a total number of combinations of $V - 4$
\end{itemize}

I now define a function that estimates a lower bound on the number of valid
combinations between these two sets for a given vertex count $V$. We denote
this conservative estimate by $\mathrm{Combinations}_{\mathrm{lb}}(V)$. This count
reflects how many internally and externally compatible sets of structures
exist, without specifying the actual types of those structures (face-types and
vertex-types). Refining this lower bound by characterizing the persistence of
internal multiplicity as a function of $S$ remains an open combinatorial
problem.

\subsection{Conservative Estimate for Structure Pair Combinations}

\[
\mathrm{Combinations}_{\mathrm{lb}}(V)=
\begin{cases}
\frac{1}{2}V^2 - 3V + 5, & \text{if $V$ is even},\\[4pt]
\frac{1}{2}V^2 - 3V + \tfrac{9}{2}, & \text{if $V$ is odd}.
\end{cases}
\]

\paragraph{Derivation.}
Fix a vertex count $V\ge 4$. The external flatness parameter $S$ ranges over
\[
S=0,1,\dots,S_{\mathrm{ub}}(V),
\qquad
S_{\mathrm{ub}}(V)=3V-6-\left\lceil \frac{3V}{2}\right\rceil
=
\begin{cases}
\frac{3V}{2}-6,& V\ \text{even},\\[2pt]
\frac{3V}{2}-\frac{13}{2},& V\ \text{odd}.
\end{cases}
\]
Within the successive peeling procedure considered here, the internal
tetrahedral count follows the construction-based range
\[
T_{\min}(V)=V-3
\qquad\text{to}\qquad
T_{\max}^{\mathrm{peel}}(V,S)
=
\max\{V-3,\;2V-8-S\}.
\]
The latter is pathway-dependent and is not asserted as a global upper bound
over all tetrahedralizations of the same boundary.

\begin{remark}[Dependence on decomposition pathway]
The quantity $T_{\max}^{\mathrm{peel}}(V,S)$ concerns the restricted
successive-peeling family defined in Section~3.6. It is not asserted as a
maximum over unrestricted tetrahedralizations of the same boundary. The cube
provides the simplest distinction: the five-tetrahedron decomposition is
generated by direct peeling, whereas the six-tetrahedron prism construction
requires a preliminary repartitioning of the remainder and lies outside the
restricted family.

The counting formula below therefore describes a conservative baseline generated by the peeling pathway.
\end{remark}

Thus, at the level of global internal invariants, the number of admissible
integer values of $T$ compatible with a given $S$ is at most
\[
T_{\max}^{\mathrm{peel}}(V,S)-T_{\min}(V)+1
=
\max\{V-4-S,\;1\}.
\]
In particular, as $S$ increases, the admissible $T$-range shrinks linearly
until it reaches the minimum value $T_{\min}=V-3$.
At $S=0$ (fully triangulated boundary), there exist polyhedra that realize the
full ladder of internal invariant values $(T,N_i,S_i)$, with the
bipyramid--prism family attaining the maximal value $T_{\max}$.

To obtain a conservative macro-level pairing count, we require only that at
least one internally compatible invariant triple exists for each admissible
value of $S$, even after the $T$-interval collapses. Writing this as a capped
count gives
\[
N_{\mathrm{int}}(V\mid S)
=
\max\{\,V-4-S,\;1\,\}
=
\max\{\,2V-8-S-(V-3)+1,\;1\,\}.
\]
Therefore the total number of compatible external--internal macro-pairs is
\[
\sum_{S=0}^{S_{\mathrm{ub}}} N_{\mathrm{int}}(V\mid S)
=
\sum_{S=0}^{V-5}(V-4-S)
+
\sum_{S=V-4}^{S_{\mathrm{ub}}} 1.
\]
The first term is a triangular sum,
\[
\sum_{S=0}^{V-5}(V-4-S)
=
(V-4)+(V-5)+\cdots+1
=
\frac{(V-4)(V-3)}{2},
\]
and the second term contributes a linear tail of length
$S_{\mathrm{ub}}-(V-4)+1$.
Substituting the parity form of $S_{\mathrm{ub}}$ and simplifying yields
\[
\mathrm{Combinations}_{\mathrm{lb}}(V)=
\begin{cases}
\frac{1}{2}V^2-3V+5, & V \text{ even},\\[4pt]
\frac{1}{2}V^2-3V+\tfrac{9}{2}, & V \text{ odd}.
\end{cases}
\]
Here the quadratic leading term arises from summing the linearly shrinking
internal ladder across $S$, while any multiplicity of realizations for a given
$(V,S,T)$ affects only lower-order corrections beyond this macro-invariant
count.

It is worth emphasizing again that this estimate is intentionally conservative.
As illustrated in Table~\ref{tab:V8} for $V=8$, multiple distinct internal
configurations may persist even after the nominal $T$-interval has ceased
to shrink linearly with $S$. Such persistence, if confirmed, increases the true number of compatible structure pairs beyond the value predicted by
$\text{Combinations}_{\mathrm{lb}}(V)$ and indicates that the present formula provides a
lower bound rather than an exact count.

\paragraph{Tail persistence.}
Whether multiple distinct realizations remain possible within the successive-peeling convention after the nominal ladder contracts, and whether such multiplicity increases with $V$, remains open.

Numerical values of $\mathrm{Combinations}_{\mathrm{lb}}(V)$ for $4 \le V \le 20$,
together with a graphical illustration of its quadratic growth, are given in
Appendix~E. Notably, this lower-bound count grows only polynomially in $V$,
whereas the number of admissible face-type configurations grows
combinatorially, as it is controlled by the cumulative sum of integer partition
counts $\sum_{S=0}^{S_{\mathrm{ub}}} p(S)$, as elaborated in the following section.

When attempting to uncover the actual face-type and vertex-type configurations
corresponding to a specific polyhedron defined by vertex count $V$,
triangulation segment count $S$, and tetrahedron count $T$, the procedures
described in Section~7 should be applied.
These methods allow one to solve for specific polygonal face arrangements and
converge to vertex degree structures that are compatible with a given global
combinatorial configuration.

It is important to emphasize that satisfying all combinatorial constraints
does not necessarily guarantee spatial realizability of all proposed
configurations. This limitation is most apparent in the vertex-type integer
combinations but may extend to other aspects of the framework as well. Some
solutions that are valid in integer space may still be non-embeddable in
three-dimensional geometry due to overlap or topological obstruction.
Realizability remains a geometric problem. One that combinatorics alone cannot
always resolve.

\subsection{Estimating the Upper Bound from Flatness and Face Combinations}

While previous sections emphasize compatibility between internal and external
combinatorics, it is also useful to estimate an upper bound of possible
polyhedral configurations based solely on the total surface angle budget. In
contrast to a full combinatorial analysis of vertex and edge constraints, this
method estimates the number of ways one can arrange polygonal faces on a
polyhedral surface using only the flatness range of the surface. It does not
take into account local realizability constraints, such as vertex incidence or
internal fillability by tetrahedra. As it doesn't account for vertex types, it
serves as a partial initial upper-bound estimate of surface complexity.

\subsection{Method Outline}

\tmtextbf{Step 1: Compute the flatness range}

Define the triangulation segment budget $S$, where $S = 0$ corresponds to a
surface composed entirely of triangles (i.e., zero extra flatness beyond
minimal curvature), and larger values of $S$ correspond to the inclusion of
increasingly flat faces (e.g., quadrilaterals, pentagons). The minimum is
always $S_{\mathrm{lb}} = 0$, corresponding to the fully triangulated case. The
maximum possible number of triangulation-equivalent segments $S_{\mathrm{ub}}$
depends on the number of vertices $V$ and follows the previously mentioned
parity--sharp linear functions:
\[ S_{\mathrm{ub}} = \left\{ \begin{array}{ll}
     \frac{3}{2} V - \tfrac{13}{2} & \text{if } V \text{is odd}\\
     \frac{3}{2} V - 6 & \text{if } V \text{is even}
   \end{array} \right. \]

\tmtextbf{Step 2: Enumerate face-type combinations using integer partitions}

Each valid configuration of polygonal face types can be expressed as an
integer partition of $S$, where each polygon type contributes a fixed number
of ``excess segments'' relative to a triangle:
\begin{itemize}
  \item A quadrilateral contributes 1,
  
  \item A pentagon contributes 2,
  
  \item And so on.
\end{itemize}
Thus, every value of $S \in [0, S_{\mathrm{ub}}]$ corresponds to a number of integer
partitions under these constraints. The total number of valid polygonal
configurations at a given $V$ is then equal to the sum of these partition
counts from $S_{\mathrm{lb}}$ to $S_{\mathrm{ub}}$. If no restriction is placed on the
number of polygon sides, then this count becomes equivalent to the
unrestricted partition function $p (S)$, a classical object in number theory
that enumerates the number of ways to write $S$ as a sum of positive integers
(ignoring order) \cite{andrews1976}. The partition function begins:
\[ p (S) = \{1, 1, 2, 3, 5, 7, 11, 15, 22, 30, 42, \ldots\} \]
Therefore, for a given vertex count $V$, the number of distinct face-type
combinations is:
\[ \sum_{S = 0}^{S_{\mathrm{ub}}} p (S) \]
For instance, for $V = 5$, I compute $S_{\mathrm{ub}} = 1$, and the number of valid
configurations is:
\[ p (0) + p (1) = 1 + 1 = 2 \]

\paragraph{\tmtextbf{Step 3: Compute total number of face-type combinations}}

Sum all partition counts (unrestricted) across the flatness range. The
resulting total gives an upper bound on the number of distinguishable surface
configurations for a polyhedron with $V$ vertices, assuming only the flatness
constraint.

\begin{figure}[H]
\centering

\begin{minipage}[t]{0.53\textwidth}
\vspace{0pt}
\centering
\captionof{table}{Lower and upper bounds on total external flatness $S$ and the
resulting upper bound on admissible face-type combinations.}
\label{tab:face_type_upperbounds}

\small
\setlength{\tabcolsep}{3.5pt}
\renewcommand{\arraystretch}{0.9}

\begin{tabular}{cccc}
\toprule
$V$ & $S_{\mathrm{lb}}$ & $S_{\mathrm{ub}}$ &
\shortstack{Face-type combinations\\(upper bound)} \\
\midrule
4  & 0 & 0  & 1    \\
5  & 0 & 1  & 2    \\
6  & 0 & 3  & 7    \\
7  & 0 & 4  & 12   \\
8  & 0 & 6  & 30   \\
9  & 0 & 7  & 45   \\
10 & 0 & 9  & 97   \\
11 & 0 & 10 & 139  \\
12 & 0 & 12 & 272  \\
13 & 0 & 13 & 373  \\
14 & 0 & 15 & 684  \\
15 & 0 & 16 & 915  \\
16 & 0 & 18 & 1597 \\
17 & 0 & 19 & 2087 \\
18 & 0 & 21 & 3506 \\
19 & 0 & 22 & 4508 \\
20 & 0 & 24 & 7338 \\
\bottomrule
\end{tabular}
\end{minipage}
\hfill
\begin{minipage}[t]{0.43\textwidth}
\vspace*{4.2cm}
\centering
\includegraphics[width=\linewidth]{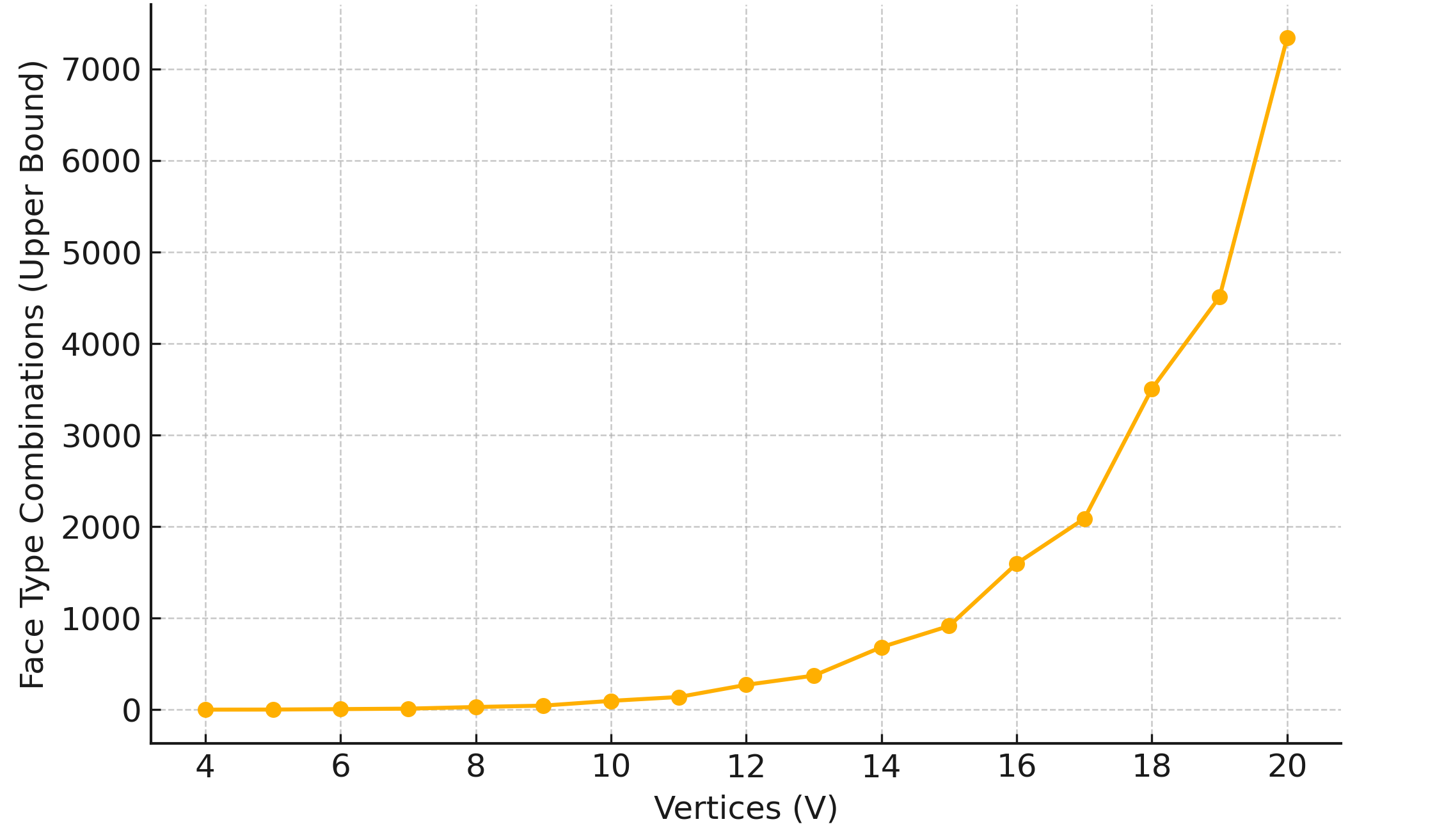}
\captionof{figure}{Upper bound of face-type combinations as a function of
vertex count.}
\label{fig:face_type_upperbounds}
\end{minipage}

\end{figure}

This curve provides an upper bound on the number of surface configurations
derivable from face configurations for each vertex count $V$, assuming only
the flatness constraint is satisfied. It does not account for limitations
imposed by vertex-type configurations, or tetrahedral fillability. This is
demonstrated in {\tmem{Appendix~A}}, as different vertex-type configurations
can exist for the same face configuration. Nonetheless, it offers a fast and
interpretable way to explore surface complexity based solely on flatness and
consequent face composition.

After determining the number of necessary triangles to achieve closure at a
desired vertex count, each potential face-type configuration generated via
this enumeration can then be filtered or corrected by the stricter
combinatorial criteria, such as the one outlined in Section~6 ("A
Combinatorial Method for Testing 3D Polygonal Enclosure and Polyhedral
Feasibility"). Accordingly, adjustments, such as the addition or removal of
triangular faces, can be made to align each configuration with more realizable
polyhedral enclosures.

\paragraph{\tmtextbf{Step 4: Refine with restricted partitions}}

If a maximum polygon side count $k_{\max}$ is imposed (such as limiting the
use of hexagons or excluding higher-order polygons) the enumeration becomes a
case of \tmtextit{restricted integer partitions}, where only parts of size
$\leq k_{\max} - 3$ are allowed. This constraint reflects that each
non-triangular face contributes a specific number of excess triangulation
segments. Restricted partitions offer a more realistic estimate by excluding
configurations that rely on geometrically unlikely face types. For example,
limiting the count of hexagons in some configurations prevents overcounting
face combinations that cannot physically tile a polyhedral surface. By summing
the number of valid restricted partitions for each value of $S = 0$ to
$S_{\mathrm{ub}}$, I obtain a tighter upper bound on the number of feasible
face-type combinations, filtering out many false positives that are
combinatorially valid but geometrically unconstructible.

\section{Discussion: Interpretation, Limits, and Future Directions}

The author declines to discuss these observations because he is not a mathematician and lacks a necessary familiarity with the literature and the field. Earlier versions included Discussion-attempts because of an enthusiasm and a sense of possibility that no longer exist. Perhaps if ever noticed or found useful, these observations could find a plausible interpretation and a direction.

\newpage

\section*{Code availability}

A reference implementation of the symbolic enclosure framework developed in this paper
is provided in the R package \texttt{polyenclose}.
The package implements the symbolic feasibility checks, external and internal bounds, and
face-type enumeration procedures described throughout the paper.
It computes \emph{necessary combinatorial conditions only} and does not attempt
geometric realization or construction.

The code is available at:
\[
\texttt{https://github.com/MoustaphaItani/polyenclose}
\]

\section{Acknowledgments}
The author has no formal mathematical training and, despite seeking guidance over fifteen years ago, did not receive the mentorship or institutional support needed to develop these ideas at the time. Mental health difficulties further complicated the ability to engage in collaborative research across unfamiliar disciplines, and financial constraints prevented access to supplementary education. The underlying constructions, bookkeeping, drawings, and exploratory ideas in this document originate entirely from the author’s independent observations and combinatorial reasoning, developed between his early teenage years and early twenties.

The author thanks Mr. Mohammad Chalhoub for providing programming support to produce the table of restricted integer partitions (prior to the author realizing the sequence had long been established), and Mr. Amr Tamimi for introducing him in earlier years to Excel, both contributions made many years ago. More recently, the author thanks Mr. Tiago Monteiro-Henriques and Ms. Hajar Assaad for briefly reviewing the manuscript, and Ms. Federica Manca and Ms. Anu Virtanen for assistance with illustrations.

\bibliographystyle{plainnat}
\bibliography{references}
\end{document}

\begin{center}
{\Large Supplementary Material for}\\[2pt]
{\Large \emph{Symbolic Constraints in Polyhedral Enclosure and Tetrahedral Decomposition in Genus-0 Polyhedra}}\\[6pt]
Moustapha Itani
\end{center}

\medskip
This document contains appendices and worked examples referenced in the main text.

\appendix

\section{Empirical Worked Examples ($V=4$--$7$)}
\label{supp:empirical}

All results in this appendix are empirical examples obtained by manual
construction. No general existence or completeness claim is made.

\begin{itemize}
  \item $N$: total surface angle units ($180^{\circ}$ each)
  \item $S$: external triangulation segments
  \item $F$: faces; $E$: edges
\item Face counts: $F_k$ denotes the number of $k$-gonal faces
      ($F_3$ triangles, $F_4$ quadrilaterals, $F_5$ pentagons, $F_6$ hexagons)
\item Vertex valencies: $V_k$ denotes the number of vertices of degree $k$
      ($V_3$ through $V_6$)
  \item $T$: tetrahedra; $N_i$: internal gluing triangles; $S_i$: internal triangulation segments
\end{itemize}

\begin{table}[H]
\centering
\caption{Empirical worked configurations for polyhedra with $V=4$ to $V=7$.
Each row represents a manually constructed example. The final two columns
verify Euler's characteristic and the internal consistency relation
$T - N_i + S_i = 1$.}
\label{tab:appendix_empirical_V4to7}
\setlength{\tabcolsep}{5pt}
\renewcommand{\arraystretch}{0.85}
\begin{tabular}{
  c c
  c c c
  c c c c
  c c c c
  c c c
  c c
}
\toprule
$V$ & $N$ &
$S$ & $F$ & $E$ &
$F_3$ & $F_4$ & $F_5$ & $F_6$ &
$V_3$ & $V_4$ & $V_5$ & $V_6$ &
$T$ & $N_i$ & $S_i$ &
$V{-}E{+}F$ & $T{-}N_i{+}S_i$ \\
\midrule
4 & 4 &
0 & 4 & 6 &
4 & -- & -- & -- &
4 & -- & -- & -- &
1 & 0 & 0 &
2 & 1 \\
\midrule
5 & 6 &
0 & 6 & 9 &
6 & 0 & -- & -- &
2 & 3 & -- & -- &
2 & 1 & 0 &
2 & 1 \\
  &   &
1 & 5 & 8 &
4 & 1 & -- & -- &
4 & 1 & -- & -- &
2 & 1 & 0 &
2 & 1 \\
\midrule
6 & 8 &
0 & 8 & 12 &
8 & 0 & 0 & -- &
2 & 2 & 2 & -- &
3 & 2 & 0 &
2 & 1 \\
  &   &
0 & 8 & 12 &
8 & 0 & 0 & -- &
0 & 6 & 0 & -- &
4 & 4 & 1 &
2 & 1 \\
  &   &
1 & 7 & 11 &
6 & 1 & 0 & -- &
2 & 4 & 0 & -- &
3 & 2 & 0 &
2 & 1 \\
  &   &
1 & 7 & 11 &
6 & 1 & 0 & -- &
3 & 2 & 1 & -- &
3 & 2 & 0 &
2 & 1 \\
  &   &
2 & 6 & 10 &
4 & 2 & 0 & -- &
4 & 2 & 0 & -- &
3 & 2 & 0 &
2 & 1 \\
  &   &
2 & 6 & 10 &
5 & 0 & 1 & -- &
5 & 0 & 1 & -- &
3 & 2 & 0 &
2 & 1 \\
  &   &
3 & 5 & 9 &
2 & 3 & 0 & -- &
6 & 0 & 0 & -- &
3 & 2 & 0 &
2 & 1 \\
\midrule
7 & 10 &
0 & 10 & 15 &
10 & 0 & 0 & 0 &
2 & 3 & 0 & 2 &
4 & 3 & 0 &
2 & 1 \\
  &    &
0 & 10 & 15 &
10 & 0 & 0 & 0 &
2 & 2 & 2 & 1 &
4 & 3 & 0 &
2 & 1 \\
  &    &
0 & 10 & 15 &
10 & 0 & 0 & 0 &
1 & 3 & 3 & 0 &
5 & 5 & 1 &
2 & 1 \\
  &    &
0 & 10 & 15 &
10 & 0 & 0 & 0 &
0 & 5 & 2 & 0 &
6 & 7 & 2 &
2 & 1 \\
  &    &
1 & 9 & 14 &
8 & 1 & 0 & 0 &
3 & 2 & 1 & 1 &
4 & 3 & 0 &
2 & 1 \\
  &    &
1 & 9 & 14 &
8 & 1 & 0 & 0 &
3 & 1 & 3 & 0 &
4 & 3 & 0 &
2 & 1 \\
  &    &
1 & 9 & 14 &
8 & 1 & 0 & 0 &
2 & 3 & 2 & 0 &
4 & 3 & 0 &
2 & 1 \\
  &    &
1 & 9 & 14 &
8 & 1 & 0 & 0 &
1 & 5 & 1 & 0 &
5 & 5 & 1 &
2 & 1 \\
  &    &
2 & 8 & 13 &
7 & 0 & 1 & 0 &
3 & 3 & 1 & 0 &
4 & 3 & 0 &
2 & 1 \\
  &    &
2 & 8 & 13 &
6 & 2 & 0 & 0 &
3 & 3 & 1 & 0 &
4 & 3 & 0 &
2 & 1 \\
  &    &
3 & 7 & 12 &
6 & 0 & 0 & 1 &
6 & 0 & 0 & 1 &
4 & 3 & 0 &
2 & 1 \\
  &    &
3 & 7 & 12 &
5 & 1 & 1 & 0 &
4 & 3 & 0 & 0 &
4 & 3 & 0 &
2 & 1 \\
  &    &
3 & 7 & 12 &
4 & 3 & 0 & 0 &
4 & 3 & 0 & 0 &
4 & 3 & 0 &
2 & 1 \\
  &    &
4 & 6 & 11 &
3 & 2 & 1 & 0 &
6 & 1 & 0 & 0 &
4 & 3 & 0 &
2 & 1 \\

\bottomrule
\end{tabular}
\end{table}

No claim is made that the listed configurations exhaust all admissible cases
for the given vertex counts.

\section{Topological Comparison of Tetrahedralization Strategies with
and without Interior Steiner Points}
\label{supp:strategies}

The different cube tetrahedralizations considered here are more naturally
compared using the general Euler relation for a tetrahedralized $3$--ball.
Let $T$ denote the number of tetrahedra, $N_i$ the number of internal
triangular faces, $E_i$ the number of interior edges, and $V_i$ the number
of interior vertices. Then
\[
T-N_i+E_i-V_i=1.
\]

Under the normal-form assumptions used in the main text, no interior
vertices are present, so $V_i=0$, and the internal triangulation segments
are precisely the interior edges, so $S_i=E_i$. The general identity then
reduces to
\[
T-N_i+S_i=1.
\]

An interior Steiner point therefore does not cause the topological identity
to fail. Rather, it moves the tetrahedralization outside the normal form in
which the reduced expression $T-N_i+S_i=1$ is used.

\begin{table}[H]
\centering
\caption{Comparison of three cube tetrahedralizations using the general
$3$--ball identity. The central-Steiner construction lies outside the
normal-form assumptions because it contains an interior vertex.}
\label{tab:appendix_steiner_comparison}
\setlength{\tabcolsep}{7pt}
\renewcommand{\arraystretch}{0.95}
\begin{tabular}{lccccc}
\toprule
Configuration
& $T$ & $N_i$ & $E_i$ & $V_i$
& $T-N_i+E_i-V_i$ \\
\midrule
Minimal five-tetrahedron decomposition
& 5 & 4 & 0 & 0 & 1 \\

Six-tetrahedron decomposition
& 6 & 6 & 1 & 0 & 1 \\

Central-Steiner twelve-tetrahedron decomposition
& 12 & 18 & 8 & 1 & 1 \\
\bottomrule
\end{tabular}
\end{table}

For the five-tetrahedron decomposition, no interior edge or interior vertex
is introduced. In the six-tetrahedron decomposition, the tetrahedra share a
single interior edge joining two existing boundary vertices, so the
normal-form specialization applies with $S_i=E_i=1$.

In the twelve-tetrahedron construction, the center of the cube is introduced
as an interior Steiner vertex and joined to the eight boundary vertices.
Thus $V_i=1$ and $E_i=8$. With $T=12$ and $N_i=18$, the general identity
gives
\[
12-18+8-1=1.
\]

The distinction is therefore not between tetrahedralizations for which the
topological relation holds and those for which it fails. The distinction is
between tetrahedralizations that satisfy the normal-form restrictions of the
main framework and those that introduce additional interior vertices.

\section{Invariance of the Proposed Identity for Polyhedra with Surface-Embedded Edge Steiner Points}
\label{supp:invariance}

To illustrate the invariance of the identity $T - N_i + S_i = 1$ under different
tetrahedral decompositions, including those involving Steiner points, consider
a single tetrahedron with one Steiner point placed on each of three edges,
each edge chosen such that the three Steiner points lie on the same face.

{\tmstrong{Decomposition A (Plane Cut Through Three Steiner Points):}}\\
Consider the plane passing through the three Steiner points. This plane
divides the tetrahedron into two polyhedral regions:
\begin{itemize}
  \item One region is a tetrahedron.
  
  \item The other is a five-faced polyhedron bounded by three quadrilateral
  faces and two triangular faces, with six vertices in total.
\end{itemize}
This five-faced polyhedron can be further decomposed into {\tmstrong{three
tetrahedra}}, yielding a total of {\tmstrong{four tetrahedra}}. The tetrahedra
are glued pairwise along {\tmstrong{three internal triangular faces}}, and the
only Steiner points used are the original three placed on the edges.

Thus, the identity is satisfied:
\[ T = 4, N_i = 3, S_i = 0 \Rightarrow T - N_i + S_i = 4 - 3 + 0 = 1 \]

{\tmstrong{Decomposition B (Edge-Based Slicing Through Steiner Points):}}\\
Alternatively, one can sequentially slice the tetrahedron through one Steiner
point at a time, joining each to the two original vertices that define its
edge, and to one of the remaining original vertices, forming tetrahedra that
share triangular faces. This approach also produces {\tmstrong{four
tetrahedra}}, glued across the same number of internal triangular interfaces
({\tmstrong{three}}), and using the same {\tmstrong{three Steiner points}}.

Again, the identity holds:
\[ T - N_i + S_i = 4 - 3 + 0 = 1 \]
This example shows that {\tmstrong{despite the different topology of the
intermediate polyhedral shapes}}, the identity is preserved, emphasizing its
{\tmstrong{combinatorial invariance}} across Steiner-supported decompositions.

\section{Worked Examples Demonstrating the Combinatorial Method
for Testing 3D Polygonal Enclosure and Feasibility}
\label{supp:examples}

This appendix presents worked examples that demonstrate the combinatorial
method. It includes three examples that are combinatorially and geometrically
valid and one example that is too flat to enclose.

\subsection*{Worked Example: Four Triangles and One Quadrilateral (Square
Pyramid)}

\begin{itemize}
  \item Step 1: Polygons: 1 square ($F_4 = 1$), 4 triangles ($F_3 = 4$)
  \item Step 2: $N = 6$
  \item Step 3: $V = \frac{6}{2} + 2 = 5$
  \item Step 4: $S = 1$
  \item Step 5: $F = 5$
  \item Step 6: $E = 3 \cdot (5 - 2) - 1 = 8$
  \item Step 7: Euler check: $5 - 8 + 5 = 2$
  \item Step 8: $S_{\mathrm{ub}} = \frac{3}{2} \cdot 5 - \tfrac{13}{2} = 1$;
  $S_{\mathrm{diff}} = 0$ (passes flatness test)
\end{itemize}

\subsection*{Worked Example: Six Squares (The Cube)}

\begin{itemize}
  \item Step 1: Polygons: 6 squares or quadrilaterals ($F_4 = 6$)
  \item Step 2: $N = (4 - 2) \cdot 6 = 12$
  \item Step 3: $V = \frac{12}{2} + 2 = 8$
  \item Step 4: $S = (4 - 3) \cdot 6 = 6$
  \item Step 5: $F = 6$
  \item Step 6: $E = 3 \cdot (8 - 2) - 6 = 12$
  \item Step 7: Euler check: $8 - 12 + 6 = 2$
  \item Step 8: $S_{\mathrm{ub}} = \frac{3}{2} \cdot 8 - 6 = 6$;
  $S_{\mathrm{diff}} = 0$ (passes flatness test)
\end{itemize}

\subsection*{Worked Example: 12 Pentagons (The Dodecahedron)}

\begin{itemize}
  \item Step 1: Polygons: 12 pentagons ($F_5 = 12$)
  \item Step 2: $N = (5 - 2) \cdot 12 = 36$
  \item Step 3: $V = \frac{36}{2} + 2 = 20$
  \item Step 4: $S = (5 - 3) \cdot 12 = 24$
  \item Step 5: $F = 12$
  \item Step 6: $E = 3 \cdot (20 - 2) - 24 = 30$
  \item Step 7: Euler check: $20 - 30 + 12 = 2$
  \item Step 8: $S_{\mathrm{ub}} = \frac{3}{2} \cdot 20 - 6 = 24$;
  $S_{\mathrm{diff}} = 0$ (passes flatness test)
\end{itemize}

\subsection*{Worked Example: 36 Hexagons}

\begin{itemize}
  \item Step 1: Polygons: 36 hexagons ($F_6 = 36$)
  \item Step 2: $N = 144$
  \item Step 3: $V = \frac{144}{2} + 2 = 74$
  \item Step 4: $S = 108$
  \item Step 5: $F = 36$
  \item Step 6: $E = 3 \cdot (74 - 2) - 108 = 108$
  \item Step 7: Euler check: $74 - 108 + 36 = 2$
  \item Step 8: $S_{\mathrm{ub}} = \frac{3}{2} \cdot 74 - 6 = 105$;
  $S_{\mathrm{diff}} = -3$ (fails flatness test)
\end{itemize}
\begin{tmindent}
  \textbf{Conclusion: }A set of 36 hexagons cannot produce an enclosed
  polyhedron.
\end{tmindent}

\section{Numerical values of $\mathrm{Combinations}_{\mathrm{lb}}(V)$}
\label{supp:combinations}

\begin{figure}[H]
\centering

% ---- Table ----
\begin{minipage}[t]{0.40\textwidth}
\vspace{0pt}
\centering

\captionof{table}{Lower-bound estimates of admissible external--internal
macro-configuration pairs, $\mathrm{Combinations}_{\mathrm{lb}}(V)$.}
\label{tab:V4to20_counts}

\small
\setlength{\tabcolsep}{7pt}
\renewcommand{\arraystretch}{0.9}

\begin{tabular}{cc}
\toprule
$V$ & \shortstack{Estimated total\\configuration sets} \\
\midrule
4  & 1   \\
5  & 2   \\
6  & 5   \\
7  & 8   \\
8  & 13  \\
9  & 18  \\
10 & 25  \\
11 & 32  \\
12 & 41  \\
13 & 50  \\
14 & 61  \\
15 & 72  \\
16 & 85  \\
17 & 98  \\
18 & 113 \\
19 & 128 \\
20 & 145 \\
\bottomrule
\end{tabular}

\end{minipage}
\hfill
% ---- Figure ----
\begin{minipage}[t]{0.56\textwidth}

% Manually adjust this value to move the figure up/down
\vspace*{2.5cm}

\centering
\includegraphics[width=\linewidth]
{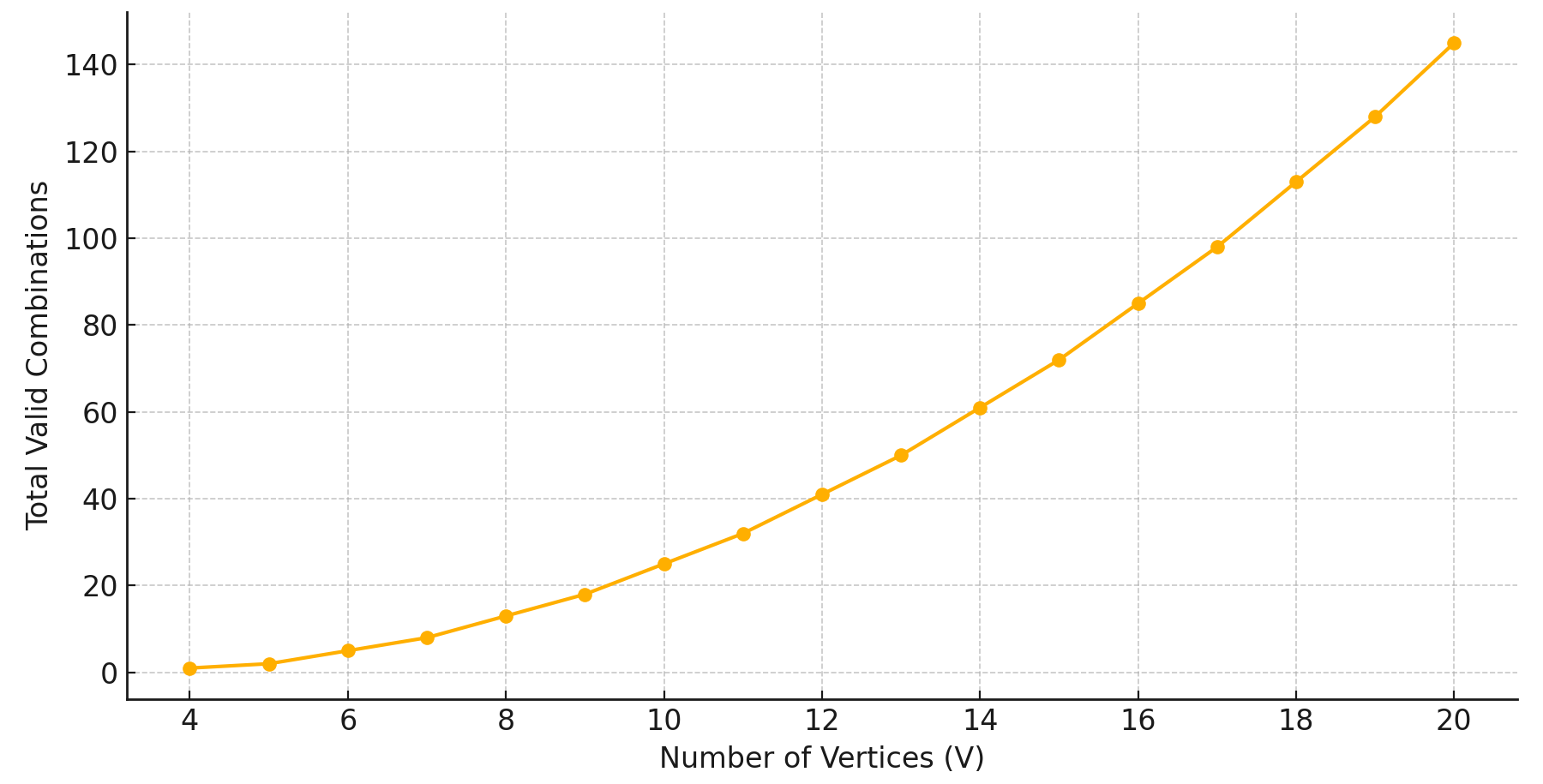}

\captionsetup{skip=4pt}
\captionof{figure}{Combinatorial sum of internal and external structure
groupings as a function of vertex count.}
\label{fig:vertex_combination_v4_v20}

\end{minipage}

\end{figure}

\section{Combinatorial Identities and Bounds}
\label{supp:identities}

This appendix records (i) \emph{proved} external and topological identities used in the main text, and (ii) the \emph{empirical/heuristic} identities and bounds collected from worked decompositions. All variables are symbolic and coordinate-free. Genus~$0$ and normal-form assumptions apply where stated. All linear ``max'' values quoted below are valid within the peeling ladder. Other decomposition approaches to triangulation can exceed these bounds.

\subsection{Proved External Identities and Bounds}

Exact external formulas for genus-0 polyhedra in normal form:
\begin{align}
  E &= 3V - 6 - S \tag{E.1}\label{appE:E1}\\
  F &= 2V - 4 - S \tag{E.2}\label{appE:E2}
\end{align}
where 
\[
S \;=\; \sum_{f} (\tmop{deg} f - 3) \quad\text{(flatness, total deviation from triangulation).}
\]

From the vertex degree constraint $\tmop{deg}(v) \ge 3$, one obtains the sharp bounds:

\medskip
\noindent
Even $V$:
\begin{align}
  0 &\le S \le \frac{3V}{2}-6, \tag{E.3a}\label{appE:E3a}\\
  \frac{3V}{2} &\le E \le 3V-6, \tag{E.4aE}\\
  \frac{V}{2}+2 &\le F \le 2V-4. \tag{E.4aF}
\end{align}

\noindent
Odd $V$:
\begin{align}
0 &\le S \le \tfrac{3V}{2} - \tfrac{13}{2} \tag{E.3b}\label{appE:E3b}\\
\tfrac{3V+1}{2} &\le E \le 3V - 6 \tag{E.4bE}\label{appE:E4bE}\\
\tfrac{V}{2} + \tfrac{5}{2} &\le F \le 2V - 4 \tag{E.4bF}\label{appE:E4bF}
\end{align}

\medskip
\noindent\textbf{Extremal Values Summary (genus-0, normal form)}  
\begin{center}
\begin{tabular}{c|c|c|c|c|c}
$V$ parity & $S_{\mathrm{ub}}$ & $E_{\mathrm{lb}}$ & $E_{\mathrm{ub}}$ & $F_{\mathrm{lb}}$ & $F_{\mathrm{ub}}$ \\
\hline
Even $V$ & $\dfrac{3V}{2} - 6$ & $\dfrac{3V}{2}$ & $3V - 6$ & $\dfrac{V}{2} + 2$ & $2V - 4$ \\
Odd $V$ & $\dfrac{3V}{2} - \tfrac{13}{2}$ & $\dfrac{3V+1}{2}$ & $3V - 6$ & $\dfrac{V}{2} + \tfrac{5}{2}$ & $2V - 4$ \\
\end{tabular}
\end{center}

\medskip
Angle-unit bookkeeping (fully triangulated boundary):
\begin{align}
  F_{\triangle} &= F + S \;=\; 2V - 4 \tag{E.5}\label{appE:E5}\\
  F_{\triangle} &= 4T - 2N_i \tag{E.6}\label{appE:E6}
\end{align}
(Equation~\eqref{appE:E6} counts triangle units: each tetrahedron contributes $4$, each internal gluing triangle cancels $2$.)

\subsection{Proved Topological Identity for 3D Complexes}

For any tetrahedralization of a $3$--ball (genus--0 volume) with $V_i$ interior vertices and $E_i$ interior edges,
\begin{align}
  T - N_i + E_i - V_i &= 1 \tag{E.7}\label{appE:E7}
\end{align}
In particular, when $V_i=0$ (normal form, no interior vertices),
\begin{align}
  T - N_i + E_i &= 1. \tag{E.8}\label{appE:E8}
\end{align}

\subsection{Boundary--Interior Euler Correspondence in Normal Form}

Using Equation~\eqref{appE:E7}, under the normal-form assumptions
$V_i=0$ and the internal triangulation segments are precisely the interior
edges, so $S_i=E_i$. Hence
\[
T-N_i+S_i=1.
\]

Since the genus--$0$ boundary satisfies
\begin{align}
  V-E+F &= 2,
\end{align}
it follows exactly that
\begin{align}
  V-E+F &= 2\,(T-N_i+S_i). \tag{E.9}\label{appE:E9}
\end{align}

Thus the boundary--interior relation is an exact consequence of the boundary
Euler characteristic together with the normal-form specialization of the
$3$--ball Euler identity. It is not an additional independent topological
invariant.

Outside normal form, the general expression
\[
T-N_i+E_i-V_i=1
\]
should be used instead of replacing $E_i-V_i$ by $S_i$.

\subsection{Core Equations}

\begin{align}
  \text{(E.9)} \quad
  & V-E+F = 2(T-N_i+S_i)
  \quad\text{(exact in normal form)}\\
  \text{(E.10)} \quad
  & T-N_i+S_i = 1
  \quad\text{(exact in normal form)}\\
  \text{(E.11)} \quad
  & N_i = 2T-V+2
  \quad\text{(exact in normal form)}\\
  \text{(E.12)} \quad
  & E-F = 2T-N_i
  \quad\text{(exact under the stated boundary assumptions).}
\end{align}

\subsection{Restricted peeling-ladder ranges (heuristic / construction--based)}

\begin{align}
  \text{(E.13r)} \quad & T_{\min} = V - 3 \\
  \text{(E.14r)} \quad & T_{\max} = 2 (V - 4) \\
  \text{(E.15r)} \quad & N_{i,\min} = V - 4 \\
  \text{(E.16r)} \quad & N_{i,\max} = 3 (V - 4) - 2 
\end{align}

\subsection{Derived Relationships (algebraic within the restricted decomposition procedure)}

The following are algebraic consequences of (E.13r)--(E.16r); they are not independent results.

\begin{align}
  \text{(E.17r)} \quad & T_{\max} = 2 N_{i,\min} \\
  \text{(E.18r)} \quad & N_{i,\max} = \tfrac{3}{2} T_{\max} - 2 \\
  \text{(E.19r)} \quad & N_{i,\max} = 2 N_{i,\min} + (V - 6) \\ 
  \text{(E.20r)} \quad & N_{i,\max} - T_{\max} = N_{i,\min} - 2 \\
  \text{(E.21r)} \quad & N_{i,\max} - T_{\max} = T_{\min} - 3 \\
  \text{(E.22r)} \quad & N_{i,\min} = T_{\min} - 1 \\
  \text{(E.23r)} \quad & 2 N_{i,\max} - 2 T_{\max} = N_{i,\min} + T_{\min} - 5 
\end{align}

\noindent\textit{Scope note.}
Equations (E.1)--(E.12) are exact under the assumptions stated for each
identity. Equations (E.13r)--(E.16r) are construction-specific ranges associated
with the restricted decomposition procedure. Equations (E.17r)--(E.23r)
are algebraic consequences of those construction-specific ranges and are not
independent results.